%% file: manuscript.tex
\documentclass[journal,10pt,draftclsnofoot,onecolumn,twoside]{IEEEtran}
\usepackage{ifpdf}
\usepackage{setspace}
\usepackage{color}
\usepackage{units}

\usepackage{ifpdf}
\ifCLASSINFOpdf
  \usepackage[pdftex]{graphicx}
  \DeclareGraphicsExtensions{.pdf,.jpeg,.png}
\else
  \usepackage[dvips]{graphicx}
  \DeclareGraphicsExtensions{.eps}
\fi

\ifCLASSOPTIONcaptionsoff
  \usepackage[nomarkers]{endfloat}
 \let\MYoriglatexcaption\caption
 \renewcommand{\caption}[2][\relax]{\MYoriglatexcaption[#2]{#2}}
\fi

\usepackage{cite}
\usepackage[cmex10]{amsmath}
\usepackage{amsthm}
\usepackage{amsfonts}
\usepackage{mathtools}
\usepackage{ctable}
\usepackage[para]{threeparttable}
\mathtoolsset{showonlyrefs}
\usepackage{amssymb}
\usepackage{bbm}
\usepackage{dsfont}
\usepackage{amscd}
\usepackage{dsfont}
\usepackage{xcolor}
\definecolor{darkgreen}{rgb}{0,0.75,0}
\usepackage{multirow,array}
\usepackage{thinsp}
\usepackage{stfloats}
\usepackage{units}
\usepackage[ddmmyyyy]{datetime}

\usepackage{hyperref}
\hypersetup{
  pdfauthor = {Mohammad M. Mansour},
  pdfcreator = {Mohammad M. Mansour},
  pdfnewwindow,
  pdffitwindow=true,
  colorlinks=true,
  linkcolor=blue,       
  citecolor=darkgreen,      
  filecolor=red,        
  urlcolor=cyan         
}

\makeatletter
\def\imod#1{\allowbreak\mkern10mu({\operator@font mod}\,\,#1)}
\makeatother

\newtheorem{theorem}{Theorem}
\newtheorem{lemma}{Lemma}

\newcommand{\nth}[1]{{#1}{\text{th}}}

\usepackage{epstopdf}

\begin{document}
\title{Inter-Frame Coding For Broadcast Communication}

\author{Hady~Zeineddine and Mohammad~M.~Mansour,~\IEEEmembership{Senior Member,~IEEE}
\thanks{H.~Zeineddine and M. M. Mansour are with the Department
of Electrical and Computer Engineering, American University of Beirut, Beirut,
Lebanon, e-mail: hma41@aub.edu.lb, mmansour@ieee.org.}}%

\markboth{J. Sel. Areas Commun. -- Recent Advances In Capacity Approaching Codes (to appear 2016)}
{Zeineddine and Mansour: Inter-Frame Coding For Broadcast Communication}

\maketitle

\vspace{-0.6in}
\input{abstract.txt}
\input{body.txt}
\vspace{-0.3in}

\singlespace\small
\bibliographystyle{IEEEtran}
\bibliography{IEEEabrv,References}

\input{Author1_Biography.txt}
\input{Author2_Biography.txt}


\end{document}

%% file: abstract.txt
\begin{abstract}
A novel inter-frame coding approach to the problem of varying channel-state conditions in broadcast wireless communication is developed in this paper; this problem causes the appropriate code-rate to vary across different transmitted frames and different receivers as well. The main aspect of the proposed approach is that it incorporates an iterative rate-matching process into the decoding of the received set of frames, such that: throughout inter-frame decoding, the code-rate of each frame is progressively lowered to or below the appropriate value, prior to applying or re-applying conventional physical-layer channel decoding on it. This iterative rate-matching process is asymptotically analyzed in this paper. It is shown to be optimal, in the sense defined in the paper. Consequently, the data-rates achievable by the proposed scheme are derived. Overall, it is concluded that, compared to the existing solutions, inter-frame coding presents a better \emph{complexity} versus \emph{data-rate} tradeoff. In terms of complexity, the overhead of inter-frame decoding includes operations that are similar in type and scheduling to those employed in the relatively-simple iterative erasure decoding. In terms of data-rates, compared to the state-of-the-art two-stage scheme involving both error-correcting and erasure coding, inter-frame coding increases the data-rate by a factor that reaches up to $1.55\times$.
\end{abstract}

%% file: body.txt
\section{Introduction}
\noindent A defining characteristic of a mobile wireless channel is the variation of the channel strength over time and frequency~\cite{FundamentalsOfWirelessCommunication}. This paper considers one implication of this channel variability in a communication scenario in which a sequence of data bits is communicated between a sender and multiple receivers. The data sequence is partitioned into a number of blocks; each block is encoded separately, using a physical- (PHY-) layer encoder, into a frame that is transmitted over the channel. The partitioning step is required because the length of the data-sequence, that is the number of data bits, is typically much larger than the maximum block-length supported by the PHY-layer encoder/decoder. As a result of the channel variability, the channel-state varies across the different transmitted frames. Subsequently, the appropriate PHY-layer code-rate, matched to the corresponding instantaneous channel-state, varies across the different frames: this is the implication dealt with in this paper. Matching the instantaneous channel-state to the appropriate code-rate, called channel-to-rate matching in this paper, is crucial to achieve high communication data-rates: if the code-rate is set to a value higher than the appropriate rate, decoding will fail at the receiver side; setting the code-rate to a value lower than the appropriate rate implies unnecessary redundancy bits are transmitted, leading to power inefficiency and data-rate loss.

In this paper, a product-coding approach is proposed to the channel-to-rate matching problem in the broadcast communication scenario. The main advantage of the proposed approach is that compared to the state-of-the-art solutions, it achieves a better \emph{complexity} versus \emph{communication data-rates} tradeoff, in the sense that is detailed in this paper. As a prelude, the state-of-the-art solutions are described next.

\subsection{Existing Solutions}
\noindent In unicast communication where the data sequence is transmitted from one sender to a single receiver, channel-to-rate matching is done frame-wise through a feedback based-scheme, e.g. as in LTE~\cite{LTE_Book}. In this scheme, instantaneous channel-state information (CSI) on the sender side is updated regularly upon feedback from the receiver, and the appropriate code-rate is set accordingly prior to transmission. When obtaining instantaneous CSI is costly or infeasible due to fast/abrupt channel-state variations, hybrid automatic repeat request (HARQ) techniques~\cite{Chase,IR} are used.
In HARQ, an encoding frame is transmitted, and then an extra retransmission is invoked each time the feedback from the receiver indicates a decoding failure. In incremental-redundancy (IR) HARQ~\cite{IR}, a retransmission
involves an ``increment" of extra redundancy bits, not transmitted previously and which corresponds to the frame. At the receiver side, the log-likelihood ratios (LLRs) of the increment are concatenated to the LLRs of the frame, and decoding is retried. Effectively, IR-HARQ increases the frame-length or equivalently decreases the corresponding code-rate upon each retransmission. The code-rate is thus decreased progressively to the appropriate rate-value where decoding succeeds. HARQ has been widely adopted in wireless communication standards (e.g. IEEE 802.16e/WiMAX\cite{802.16} and 3GPP-LTE\cite{LTE}).

In the broadcast communication scenario where the data sequence is sent from one sender to multiple receivers, the channel-to-rate matching problem has to be reconsidered because the aforementioned frame-wise feedback-based channel-to-rate matching schemes do not scale well as the number of receivers grows. The reason is that channel-state variation is two-folded: 1) per single
receiver, the channel-state varies across different frame-transmissions and 2) per single frame-transmission, the channel-state varies across different receivers. This means that for each frame, the code-rate must be matched to the worst instantaneous channel-state instance among all receivers. In the typical case where the worst channel-state instance does not correspond to the same receiver for all the transmitted frames, the following happens: for each receiver, there exist some transmitted frames which code-rates are matched to channel-state instances that are worse than the receiver's channel-state instances, and therefore, these code-rates are lower than the appropriate values corresponding to the instantaneous channel-state of this receiver. As a result, the average number of transmitted bits per frame will be larger than the average number of bits required by each receiver for successful recovery of the transmitted frames.

The state-of-the-art solution to the problem of channel-state variation in broadcast communication
is a two-stage forward-error-control scheme that combines application-layer (APP-layer) erasure coding with PHY-layer channel coding. The data sequence is encoded, using an erasure code, to form a larger sequence. The latter sequence is partitioned into blocks, each of which is subject to PHY-layer encoding forming a frame that is transmitted over the channel. On the receiver side, the frames that fail PHY-layer decoding are discarded, whereas symbols in the successfully-decoded frames are forwarded to erasure decoding. Erasure decoding is then used to recover the symbols erased due to channel decoding failures. The two-stage scheme does not solve the channel-to-rate matching problem for every frame; rather, it ensures communication is done reliably even under the PHY-layer decoding failures occurring when the code-rate is higher than the appropriate value that matches the corresponding instantaneous channel-state.

Due to their linear-time encoding and decoding and excellent coding performance, LT/Raptor codes\cite{Luby_LT,Shokrollahi_Raptor} or an enhancement of them (RaptorQ\cite{RaptorQ}),  are typically deployed as erasure codes.
Beside their rateless encoding procedure, a key factor in their choice as the deployed erasure codes is their inactivation decoding method~\cite{Inactivation}, which combines the low complexity of the belief-propagation method with the decoding guarantee of Gaussian elimination. The two-stage scheme, involving Raptor coding, has been included in the 3GPP multimedia broadcast/multicast services (MBMS)\cite{MBMS} and digital video broadcasting-handheld (DVB-H)\cite{DVB} standards.
\textcolor{black}{Other codes such as the binary deterministic rateless (BDR) codes have been
proposed to replace LT/Raptor codes as erasure codes in broadcast communication~\cite{BDR}. Variations of the two-stage scheme, based on network coding, has also been proposed (e.g.~\cite{NetworkCoding}).}

To understand the significance of the inter-frame coding approach proposed in this paper, the advantages and drawbacks of the two-stage scheme are stated briefly next. First, the simplicity of the erasure-channel model implies that the corresponding decoding algorithms are significantly simpler than the PHY-layer decoding algorithms. Thus, erasure decoding can be efficiently implemented in software, making it flexible and adaptable to the type of application involved; yet, hardware accelerators are proposed to enhance the power and decoding throughput~\cite{Embedded}. Second, due to this combination of flexibility and relative-simplicity,  the erasure-code size is made much larger than the PHY-layer frame-length so that the erasure code spans multiple PHY-layer frames. Therefore, erasure coding is used to withstand any type of channel-state variation that leads to a relatively high PHY-layer decoding-failure rate. Third, the scheme is suitable for the broadcast communication scenario since it scales well as the number of receivers goes up.

The main drawback of the two-stage scheme is that it incurs a loss in the achievable data-rates. This is due to the underlying erasure-channel abstraction it involves: a frame which PHY-layer decoding fails is discarded, i.e. is effectively erased. This is to be contrasted with the case in IR-HARQ where, in case of decoding failure, increments of redundancy bits are sent from the transmitter and decoding is retried. This deficiency is addressed in~\cite{Solution1,Solution2,Solution3} by applying post-decoding processing on the frames which decoding fails. However, these schemes incur a high complexity overhead~\cite{Solution1} and/or impose certain assumptions on the bit-error-rate in the unsuccessfully-decoded frames, thus, limiting their applicability~\cite{Solution2,Solution3}.

\subsection{Proposed Approach}
\noindent In this paper, an approach is developed to incorporate the channel-to-rate matching step into the decoding process; this approach is called here inter-frame coding. The inter-frame encoding step generates both frames and subframes, where a subframe-length is only a small portion of the frame-length. The frames are formed using the conventional PHY-layer encoding, called here intra-frame encoding, and the subframes are formed by applying linear encoding on increments corresponding to the respective frames. The inter-frame decoding is an iterative process in which two kinds of procedures are applied: 1) the conventional PHY-layer channel decoding, called here intra-frame decoding, to recover the transmitted frames, and 2) progressive concatenation of subframes to frames that are unsuccessfully-decoded prior to retrying intra-frame decoding on them. Inter-frame decoding is then an iterative process in which unsuccessfully decoded frames are recovered by progressively decreasing their code-rates and retrying channel decoding on them; therefore, it applies a form of rate matching for each of the involved frames.

Inter-frame coding can be viewed as a variation or a subclass of the product codes introduced by Elias~\cite{ProductElias}, particularly related to the interleaved~\cite{ProductInterleaved} and irregular~\cite{ProductIrregular} subclasses of product codes. In this regard, the PHY-layer intra-frame encoding (frame-generation) can be viewed as row coding, whereas subframe generation can be viewed as a form of column coding. Overall, inter-frame coding is a product-coding approach, with peculiar encoding and decoding algorithms that are developed so that they attain certain features that make inter-frame coding advantageous compared to the previously mentioned existing solutions for the broadcast communication scenario. This is described in the next two paragraphs.

In terms of complexity, the significant feature of inter-frame decoding can be stated as follows: besides intra-frame decoding, inter-frame decoding involves nearly the same scheduling and operations of the iterative LT erasure decoding process in~\cite{Luby_LT}. This means that the previously mentioned advantages attributed to erasure decoding in the two-stage error-control scheme can also be maintained in inter-frame decoding: this includes decoding simplicity, flexibility, and ability to handle large number of frames, as well as the possibility of a multi-layer (i.e. PHY- and APP-) implementation of the proposed scheme.

In terms of coding performance, inter-frame coding can achieve better data-rates compared to the two-stage scheme: this is because an unsuccessfully-decoded frame is not discarded in inter-frame decoding; rather, its code-rate is decreased and intra-frame decoding is reapplied on it. Besides, being a coding approach, it scales well as the number of receivers grows making it suitable for the broadcast scenario.

The contribution of this paper includes algorithm development and code analysis. It can be divided into three parts.
First, the encoding and decoding algorithms are developed, and their features discussed.
Second, the rate-matching process involved in inter-frame decoding is analyzed. As a prelude, the channel-state variation in relation to inter-frame decoding is modeled using a channel-characterizing probability distribution over $\mathbb{Z^+}$.
Under such model, inter-frame decoding can be viewed as an iterative process in which rate matching is applied progressively. This model allows to set an optimality criterion for this rate-matching process, as well as an upper-bound on the data-rates that are achievable by inter-frame decoding. Then, this process is shown to be equivalent to a two-phase message-passing algorithm that is applied on a bi-partite graph, and which generalizes the LDPC erasure iterative message-passing decoding. Under some simplifying assumption on the channel-characterizing probability distribution, the two-phase message-passing procedure is asymptotically analyzed, and optimal degree distributions are constructed accordingly; these distributions characterize inter-frame codes that can achieve the previously mentioned upper-bound on the data-rates.
Third, the data-rates, shown to be achievable by inter-frame coding, are then compared to their counterparts in the state-of-the-art solutions. It is shown that inter-frame coding increases the data-rate by a factor that is dependant on the parameters of the corresponding channel-characterizing distribution; Compared to frame-wise feedback-based rate-matching processes, having a target frame-error-rate (FER) of $10^{-3}$,  this factor can reach $5\times$ for an infinitely large number of receivers. Compared to the state-of-the-art two-stage scheme, this factor can reach up to $\sim1.55\times$.

The rest of the paper is organized as follows. In Section~\ref{s:Setup}, the problem setup is described. The proposed encoding and decoding algorithms are detailed in Section~\ref{s:Algorithms}. In Section~\ref{s:Code-Analysis}, the rate-matching process is analyzed and the data-rates achievable by inter-frame coding are derived. These data-rates are compared to those achievable by the other schemes in Section~\ref{s:Comparison}. Section~\ref{s:Conclusions} concludes the paper.

\section{Problem Setup}\label{s:Setup}
\noindent The problem setup is made general enough to fit both the proposed inter-frame coding scheme and the conventional two-stage scheme. This simplifies the comparison done in Section~\ref{s:Comparison}.

Figure~\ref{f:Problem_Setup} illustrates the problem setup, on the sender side. A sequence of $N_F\!\cdot\! K$ data bits is to be sent to a number of receivers. In case of the two-stage scheme, the $N_F\!\cdot\!K$ date bits are encoded using an APP-layer erasure encoding procedure with rate $R_E=\frac{N_F}{N_T}$. In case of inter-frame coding, no APP-layer erasure encoding is applied, and $R_E$ can be thought to be equivalently $1$ and $N_T=N_F$. The resulting $(N_T\cdot K)$-bit sequence is partitioned into $N_T$ blocks. Each block consists of $K$-bits or equivalently $L$ symbols, where each symbol includes $\frac{K}{L}$ bits.
A block is forwarded to the PHY-layer channel encoder that generates a rate-$R_L$ encoded frame, where $R_L\!=\!\frac{K}{N+D\cdot\Delta}$ and $D,\Delta,N\!\in\!\mathbb{N}$ are parameters to be defined next.
This encoder, called ``intra-frame" in this paper, is designed to be rate-compatible (as in e.g.~\cite{RC_Turbo2,RC_Protograph,HWRaptor,rate_compatible_mansour}), where for any two code-rates $R\!>\!R'$, the rate-$R'$ frame is a concatenation of the rate-$R$ frame and extra redundancy bits. The rate-$R_L$ encoded frame can be viewed as a concatenation of a rate-$R_H$ $N$-bit encoded frame, where $R_H\!=\!\frac{K}{N}$, and $D$ vectors or increments of size $\Delta$ bits each. The $\nth{j}$ increment, is denoted by $\Delta(f,j)$, where $1\!\leq\! j\!\leq\! D$ and $f$ is the frame index. The sequence of the $N_T$ $(N\!+\!D\cdot\Delta)$-bit frames is then forwarded to the bit-stream generation process that produces the bit stream to be transmitted over the channel.
\begin{figure*}[t]
\centering
\includegraphics[scale=0.58]{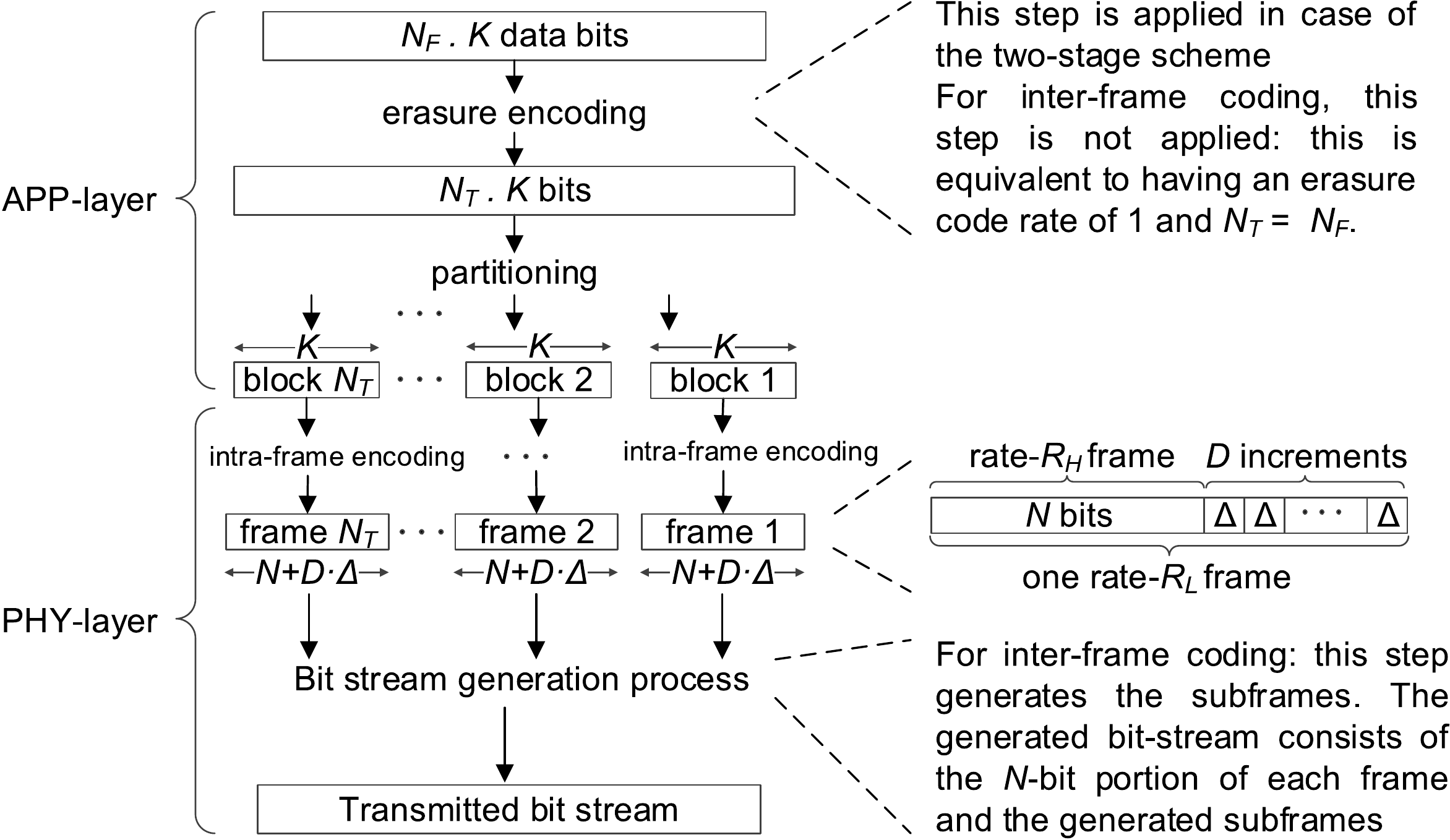}
\caption{Encoding procedure setup}
\label{f:Problem_Setup}
\end{figure*}

At the receiver side, the log-likelihood ratios (LLRs) of the transmitted bits are input to a recovery process to restore the $N_T$ blocks. This process is the focus of this paper. The following related terminology is defined. \textcolor{black}{The frame error-rate (FER) of the recovery process is the probability that the process fails to recover a randomly chosen frame of the $N_F$ frames.}
An $N$-LLR vector corresponding to the first $N$-bit portion of a transmitted frame $f$, $f\leq N_F$, is denoted by $\Lambda_N(f)$, whereas a $\Delta$-LLR vector corresponding to $\Delta(f,j)$ is denoted $\Lambda_\Delta(f,j)$. In the paper, unless otherwise indicated, indexing is done by including the corresponding index/indices in brackets $(\cdot)$.

\textcolor{black}{
A note should be made on the partitioning of the bit sequence into $K$-bit blocks in the setup.
A plausible solution to the cross-frame channel-state variation is to encode the
$N_F\cdot K$ data bits into one frame at the PHY-layer and transmit it over the channel. The motivation is that the transmission of this single large frame will span a relatively long time; therefore, the optimal code-rate of this frame can be deduced
from statistical channel-state information (CSI), as is suggested by the basic communication theory of Shannon. This solution however is impractical because the intra-frame decoder, usually implemented in hardware as an ASIC (Application-Specific
Integrated-Circuits) solution for power and throughput reasons, has to support
now much larger frame-lengths as the information block-length is multiplied by
a factor of $N_F$ (from $K$ to $N_F \cdot K$). This is clearly impractical in virtually any
hardware-based system. In contrast, the proposed inter-frame decoding scheme has two components: 1) an intra-frame decoding process which corresponding decoder needs to support a maximum frame-length of $N+D\cdot\Delta$ bits only, and 2) an iterative process that is similar to the simple LT erasure iterative decoding. The relative simplicity of this latter process means that it can be performed mostly using hardware resources that are simple and that need not be dedicated to inter-frame decoding. The aforementioned hardware resources include both the computational units and the memory needed to store the LLRs of the unsuccessfully decoded frames and their neighbor subframes. This point will be clarified in the description of inter-frame coding in Section~\ref{s:Algorithms}.
}

\section{Encoding and Decoding Algorithms}\label{s:Algorithms}
\noindent The inter-frame encoding and decoding procedures are described next.

\subsection{Encoding}\label{s:EncProc}

\noindent Starting from an initial set of $N_T\!=\!N_F$ $K$-bit data blocks, inter-frame encoding generates the following: $N_F$
$N$-bit frames and $K_S$ $\Delta$-bit subframes, forming a total of $N_F\times N +K_S\times\Delta$ bits that are transmitted over the channel. The encoding procedure is illustrated in Fig.~\ref{f:Encoding}. Each block is intra-frame encoded into a $(N\!+\!D\cdot\Delta)$-bit frame. Generation of the $N_F$ frames, in inter-frame encoding, is straightforward: the first $N$-bit portion of each $(N\!+\!D\cdot\Delta)$-bit frame is transmitted over the channel.

The generation of the $K_S$ subframes is done according to a $K_S\!\times\! N_F$  matrix $\mathbf{H}\!=\![h_{(i,j)}]$, where $h_{(i,j)}\in\mathbb{N}$ denotes the entry in row $i$ and column $j$ in $\mathbf{H}$; $\mathbf{H}$ is the ``inter-frame" code generator matrix. Each column of $\mathbf{H}$ has a maximum of $D$ nonzero distinct positive integer entries that are less than $D+1$. The generation of subframe $s$, $s\!\leq\! K_S$, is described fully by row $s$ of $\mathbf{H}$; subframe $s$ is then formed to be equal to:
\begin{equation*}
\bigoplus_{1\leq f\leq N_F}\Delta(f,h_{(s,f)})
\end{equation*}
where $\bigoplus$ denotes vector addition $\bmod 2$ applied on $\Delta$-dimensional binary vectors, and $\Delta(f,0)$ denotes the $\Delta$-dimensional zero vector, for every $f$. Therefore, subframe $s$ is formed using the following procedure:\\
\emph{Initialize subframe $s$ to a $\Delta$-bit zero vector. For each nonzero entry $h_{(s,f)}\!=\!a$ of $\mathbf{H}$, update subframe $s$ to the output vector formed by bit-wise XORing subframe $s$ with increment $\Delta(f,a)$.
}

\begin{figure*}[t]
\centering
\includegraphics[width=\textwidth]{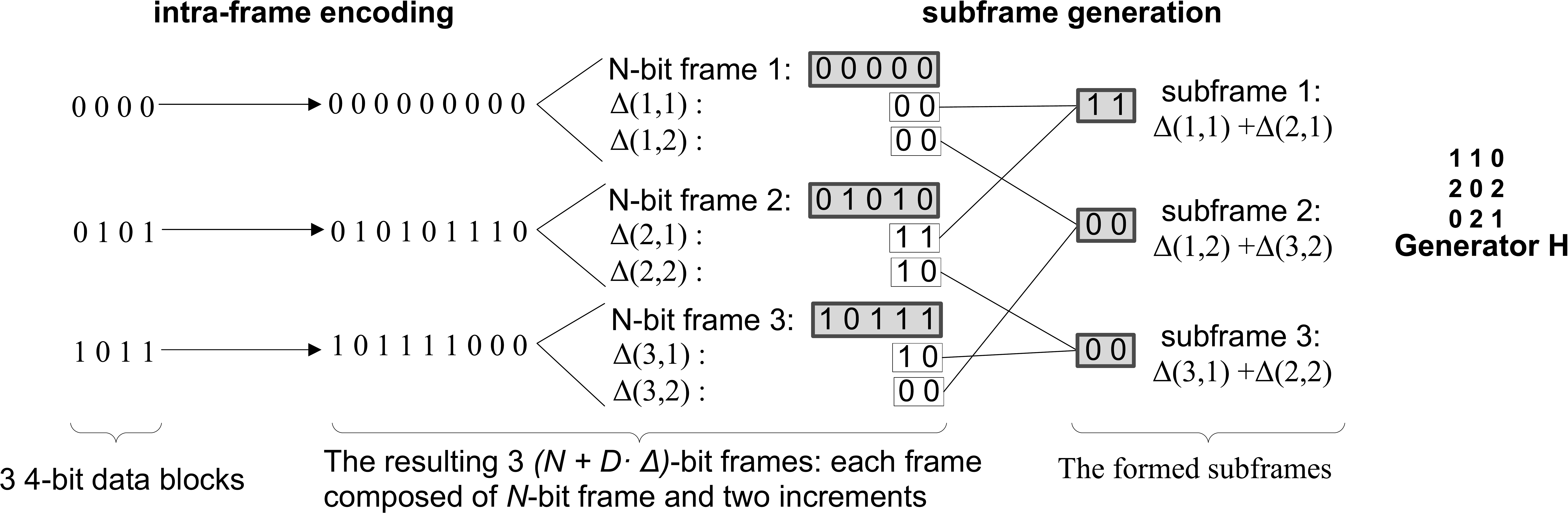}
\caption[Inter-frame encoding example]{Inter-frame encoding example: Generating $K_S\!=\!3$ subframes from $N_F\!=\!3$ frames, with parameters $(K,N,\Delta,D)\!=\!(4,5,2,2)$. The $21$ bits contained in dark-colored boxes in the figure are transmitted over the channel.}
\label{f:Encoding}
\end{figure*}

Some related terminology is defined next: frame $f$ and subframe $s$ are said to be neighbors if $h_{(s,f)}\!\neq\!0$. The degree of frame $f$ is defined as the number of non-zero entries in column $f$ of $\mathbf{H}$. The degree of subframe $s$ is defined as the number of non-zero entries in row $s$ of $\mathbf{H}$.

\subsection{Decoding}\label{s:DecProc}

\noindent {\bf Input:} It consists of the inter-frame code generator matrix $\mathbf{H}$, the intrinsic LLRs of the $N_F$ $N$-bit frames $\Lambda_N(f)$, for $f\!\leq\! N_F$, and the LLRs of the $K_S$ $\Delta$-bit subframes $\Lambda_S(s)$, for $s\!\leq\! K_S$.

\noindent {\bf Output:} It includes an $N_F$-bit flag vector indicating decoding success or failure of the $N_F$ frames. For each successfully-decoded frame, the corresponding $K$ data bits are recovered.

\noindent {\bf Procedure} The decoding procedure applies iteratively three main steps described next:
\begin{itemize}
\item [.] \emph{Intra-frame decoding:} the PHY-layer channel decoding operation applied on each frame.
\item [.] \emph{Subframe LLR update:} when intra-frame decoding of a frame $f$ succeeds, the corresponding increments $\Delta(f,i)$, $i\!=\!1,\cdots,D$, are recovered. Consider any subframe $s$ that is a neighbor of frame $f$, i.e. $h_{(s,f)}\!=\!a\neq0$: subframe $s$ is equal to $\Delta(f,a)\oplus\Psi$ where $\Psi$ is a XOR of other increments.
    $\Lambda_S(s)$ is updated now so that it corresponds to $\Psi$, as follows:
    \begin{equation}
\Lambda_S(s)\leftarrow (-1)^{\Delta(f,a)}\times \Lambda_S(s)
\end{equation}
where the $\nth{m}$ entry of  $(-1)^{\Delta(f,a)}$ is $-1$ if the $\nth{m}$ bit of $\Delta(f,a)$ is $1$, and $1$ otherwise, and $\times$ is entry-wise multiplication. Effectively, this is equivalent to removing $\Delta(f,a)$ from the list of inputs to the XOR operation forming the transmitted subframe $s$.
\item [.] \emph{Subframe-to-frame concatenation:} This step is applied when the number of inputs to a subframe $s$ is effectively reduced to one increment, this reduction being due to successive application of the previous \emph{subframe LLR update} step. Assume this single increment is $\Delta(f',b)$ for some $b\leq D, f'\leq N_F$, then $\Lambda_S(s)$ corresponds now to $\Delta(f',b)$. $\Lambda_S(s)$ is then concatenated to the LLR vector of frame $f'$, decreasing its code-rate. If decoding of frame $f'$ has failed prior to this concatenation, frame $f'$ is rescheduled for intra-frame decoding.
\end{itemize}

The last two of these three steps contribute to the recovery from possible failures in intra-frame decoding of some frames, as is clarified next. Suppose intra-frame decoding of a frame $f$ fails. As other frames are successfully decoded throughout the inter-frame decoding process, there will expectedly appear a subframe $s$ that is neighbor of the frame $f$ such that: all the neighbor frames of $s$, except $f$, are successfully-decoded and therefore recovered. The LLR-vector corresponding to subframe $s$, $\Lambda_S(s)$, is concatenated to the LLR vector corresponding to $f$. This means that the frame-length of $f$ is effectively increased, or equivalently that the code-rate corresponding to frame $f$ is decreased. Decoding is then retried whenever such concatenation happens. When a sufficient number of subframes is concatenated to $f$, that is when the code-rate corresponding to $f$ becomes low enough, the intra-frame decoding of frame $f$ will succeed, and frame $f$ will be recovered. This illustrates the view of inter-frame decoding as an iterative process that applies a form of rate matching. The detailed decoding algorithm is presented next. For simplicity of exposition, it is assumed each row of $\mathbf{H}$ has at least two non-zero entries.

{\bf Algorithm:} Three disjoint sets $\mathcal{R}$, $\mathcal{U}$, $\mathcal{P}$ are formed, where $\mathcal{R}$ includes the indices of the successfully-decoded frames, $\mathcal{U}$ the indices of the unsuccessfully-decoded frames, and $\mathcal{P}$ the indices of the frames scheduled for intra-frame decoding. Both $\mathcal{U}$ and $\mathcal{R}$ are initialized to $\varnothing$ and $\mathcal{P}$ to $\{1,\cdots,N_F\}$. This decoding procedures repeats the following iteration as long as $\mathcal{P}\!\neq\!\varnothing$:\\
\textbf{Iteration:} If $\mathcal{P}\!=\!\varnothing$, quit decoding, else, pick randomly an element $f$ from $\mathcal{P}$. Apply intra-frame decoding on the vector of LLR values corresponding to frame $f$. If intra-frame decoding is unsuccessful, move $f$ to $\mathcal{U}$. If it is successful, proceed as follows:
\begin{itemize}
\item[1.] $f$ is moved from $\mathcal{P}$ to $\mathcal{R}$.
\item[2.] \emph{Increment Recovery:} the $K$ information bits corresponding to frame $f$ are recovered and used through intra-frame encoding to obtain the increments $\Delta(f,k), k=1,\cdots,D$.
\item[3.] For every $s$ where $h_{(s,f)}=a\neq 0$:
 \begin{itemize}
 \item[a.] \emph{Subframe LLR update:} the LLR vector, $\Lambda_S(s)$, is updated: $\Lambda_S(s)\leftarrow \Lambda_S(s)\times (-1)^{\Delta(f,a)}$.
  \item[b.] Set $h_{(s,f)}$ to $0$.
  \item[c.] \emph{Subframe-to-frame concatenation:} If the number of nonzero entries of row $s$ of $\mathbf{H}$ reduces to $1$, where $h_{(s,f')}=b>0$, then: 1) if $f'\in \mathcal{U}$, $f'$ is moved to $\mathcal{P}$, that is, scheduled for intra-frame decoding retrial, and 2) $\Lambda_S(s)$ is concatenated to the LLR vector of $f'$. By abuse on notation, increment $\Delta(f',b)$, or equivalently subframe $s$, is said to be concatenated to frame $f'$.
 \end{itemize}
\end{itemize}

\subsection{Resemblance to LT erasure decoding}
\noindent The resemblance of inter-frame decoding to the LT iterative erasure decoding is due to the algorithm feature that only frames that are successfully-decoded are involved in the process of decreasing the code-rate, or equivalently incrementing the LLR-vector, of the unsuccessfully-decoded frames. Note that such feature is not confined to inter-frame coding; a similar feature is used for example in~\cite{Shortened_HARQ} to design efficient channel codes for the HARQ schemes. The resemblance between inter-frame decoding and iterative LT erasure decoding is clarified next.

In inter-frame decoding, step $3$ is invoked only when the corresponding frame $f$ is successfully decoded. This has two implications. First, the LLRs corresponding to the successfully-decoded frame $f$ are now set to $\pm\infty$; therefore the LLR update in step 3a is simplified into changing signs of the LLR-vector $\Lambda_S(s)$, equivalent to performing $\Delta$ 2-input XOR operations. Second, step 3a will only be invoked a maximum number of times that is equal to the number of nonzero entries of $\mathbf{H}$, which is $D\!\cdot\! N_F$.

The LT iterative decoding procedure in~\cite{Luby_LT} can be viewed as a series of successive edge-processing steps; processing each edge would result in effectively removing it from the code bi-partite graph. The mentioned features of inter-frame decoding have their counterparts in LT decoding. First, edge processing is invoked when an input symbol is recovered. Second, one edge in the LT bipartite graph, or equivalently one non-zero entry of the generator matrix, is processed at most once in the whole decoding procedure. Third, this edge processing consists of a number of 2-input XOR operations.

As a result, the following conclusion can be made: in addition to the conventional intra-frame decoding, inter-frame decoding involves operations that are similar, in both their scheduling and type, to operations involved in LT iterative erasure decoding.

{\bf A note on complexity:} the decoding procedure involves the following operations per frame: the XOR-operation (step 3a) is performed $D\cdot\Delta$ times, the intra-frame encoding procedure in step 2 is performed $1$ time, and the intra-frame decoding procedure is performed an average of $X\!>\!1$ times. Typically, the intra-frame decoding procedure is much more computationally intensive compared to intra-frame encoding and the $D\cdot\Delta$ XOR operations. The average number of intra-frame decoding attempts per frame, denoted here by $X$, is dependent on 1) the intra-frame decoding failure-rate and its change with the concatenation of the increments LLRs, and 2) the criteria with which a frame is selected from $\mathcal{P}$ for intra-frame decoding. The first factor is a function of the channel and the deployed intra-frame code. The second factor raises the problem of sorting the frames scheduled for intra-frame decoding according to the probability of intra-frame decoding success, given the LLRs corresponding to each frame. Therefore, this problem necessitates developing a relatively low-complexity method to estimate the probability of intra-frame decoding success of a frame without performing the computationally-intensive decoding itself. If such method exists, intra-frame decoding of a frame is done only when the method predicts that the decoding of the frame will succeed. This problem, however, is outside the scope of this paper.

\subsection{Discussion}
The main advantage of the proposed inter-frame coding scheme can be stated as follows: it pertains the implementation-friendly features of the state-of-the-art two-stage scheme, while achieving higher data-rates.

The first part of this previous statement can be deduced from the description of the inter-frame decoding algorithm. The resemblance between iterative LT erasure decoding and the proposed inter-frame decoding implies that the advantages attributed to erasure decoding in the two-stage error-control scheme can also be pertained in inter-frame decoding. This includes the simplicity of decoding: aside from the conventional intra-frame decoding, inter-frame decoding involves relatively simple operations that need no special dedicated hardware. It also includes the flexibility in the design of inter-frame codes, i.e. of the code generator matrix $\mathbf{H}$; one feature of such flexibility is the ability to handle large number, determined in real-time, of frames. The resemblance suggests also that inter-frame decoding can be efficiently implemented using a multi-layer approach, involving both the PHY-layer and APP-layer, so that the hardware-overhead of inter-frame decoding can be kept minimal.
The comparison between the two-stage scheme and inter-frame coding, in terms of complexity and implementation,
depends on the communication schemes and protocols, as well as on the system architecture and implementation. A detailed treatment of this subject is, therefore, not in the scope of the paper.

The second part of the statement, on the enhancement brought by inter-frame coding to the achievable data-rates, will be quantitatively studied in the rest of the paper.

\section{Optimality Analysis of Rate-Matching under Inter-frame Decoding}\label{s:Code-Analysis}
\noindent In this section, the rate-matching process involved in inter-frame decoding is analyzed; thus, the data-rates achievable by inter-frame coding are derived, under some assumptions on the channel and deployed intra-frame code. They are compared to the data-rates achievable by the conventional solutions in Section~\ref{s:Comparison}.

In this paper, the data-rate achievable by a scheme is measured in terms of the resulting \emph{effective frame-length} defined as: the average number of bits, per $K$-bit data block, that is required to be transmitted for the success of the recovery process\footnote{The recovery processes considered and compared in this paper (Section~\ref{s:Comparison}) are: inter-frame coding, two-stage scheme, and frame-wise feedback-based rate-matching.}. The effective frame-length is therefore equal to the total number of transmitted bits divided by $N_F$. Thus, for an inter-frame code with parameters $(N_F,K_S,N,\Delta)$, the effective frame-length is:
\begin{equation}
\frac{N_F\cdot N+ K_S \cdot \Delta}{N_F}=N\cdot\left(1+\frac{K_S}{N_F}\cdot\frac{\Delta}{N}\right).
\end{equation}
In this section, each of the $\Delta$ and $N$ values are assumed to be predefined; in addition, they are fixed across the proposed and conventional schemes. Finding the effective frame-length attained by inter-frame coding is then equivalent to finding the minimum value of $\frac{K_S}{N_F}$ that is sufficient for the inter-frame decoding process to succeed. The inter-frame decoding success criterion is formally defined in subsection~\ref{s:Bound}.

In this section, the effective frame-length is derived by analyzing the inter-frame decoding process. The analysis is preceded by setting a simplifying model of the across-frame channel variability. As a result of this modeling, inter-frame coding is ``abstracted" into an iterative rate-matching process in which the code-rate of each frame is progressively decreased to its appropriate rate. It is this rate-matching process that is analyzed in the rest of the section, with the aim of finding the minimum $\frac{K_S}{N_F}$ it requires to succeed. The performed analysis is asymptotic, meaning that it assumes $N_F$ to be infinitely large. The asymptotic nature of the analysis is motivated by three main considerations.
First, asymptotic analysis provides the performance limits of the rate-matching process. One side result is that it becomes possible to quantify the performance degradation caused by constructing inter-frame codes of relatively small sizes. Second, some complications rising in the analysis, when $N_F$ is finite, can be overlooked in the asymptotic case; an example of such complication is the existence of some correlation of the channel-state across frames. This problem can be approached by interleaving the frames or including the consecutive frames into different inter-frame codes, both techniques being unrestricted in the asymptotic case by any upper limit on $N_F$.
Third, the asymptotic performance of the rate-matching process can be well analyzed and optimized using a small set of compact mathematical expressions and inequalities.

\subsection{Channel Variability Model}\label{s:Model}
\noindent In this subsection, the variation of the channel state across frames is abstracted into a simple across-frame appropriate-rate variability model. This is done by setting some simplifying assumptions on the intra-frame code, channel, and communication scenario. The model is general enough to encompass different communication scenarios. Yet, its simplicity allows efficient analysis of inter-frame decoding by abstracting it into a rate-matching process. The assumptions underlying the model are stated as follows:

\emph{Assumption}~1: For any integer $m\geq 0$, the decoding performance of the intra-frame code, formed by concatenating the original $N$-bit frame to any set of $m$ distinct increments, equals that of a conventional code of rate $\frac{K}{N+m\cdot\Delta}$. The decoding performance here is measured in terms of the intra-frame decoding-failure rate.
This means that this decoding performance is sensitive to the number but not the indices of the concatenated increments. The design of intra-frame codes with such property is not trivial given the observation made in\cite{HARQ_Gain} that the sensitivity of the turbo decoding performance to the LLRs of different portions of the code differs.

\emph{Assumption}~2: Given the channel statistical model, the channel-states over which different frames are transmitted are assumed independent instances. This assumption has no bearing on the validity of the model. It is made here since, as briefly discussed in the introduction of this section, channel-state correlation can be overlooked in the ongoing asymptotic analysis.

\color{black}
\emph{Assumption}~3: The number of increments required to be concatenated to a frame $f$ for successful decoding depends, solely, on the channel-state when the $N$-bit portion of $f$ is transmitted. This assumption is clearly not realistic because it ignores the facts that 1) in most cases, the channel-state during the transmission of the increments concatenated to $f$ cannot be deduced from the channel-state when the $N$-bit portion is transmitted, and 2) the channel-state varies across the different increments transmission, so that concatenating a number of increments, transmitted over \emph{good} channel-state, to $f$  may be sufficient to successfully decode $f$, while concatenating a larger number of increments, transmitted over \emph{worse} channel conditions, may not be sufficient.\\
This assumption is intended primarily to simplify the channel variability model developed next in this subsection; consequently, it makes it possible to clearly identify the rate-matching process underlying inter-frame decoding and to deduce the relation between this process and iterative erasure decoding applied on bi-partite graphs. This relation will help understand, analyze, and optimize the rate-matching process as will be clear throughout the rest of the section.


\color{black}

\textbf{Model:} A received frame $f$ is characterized by an integer, $\kappa_{(f)}\geq0$, defined as the number of LLR increments, each $\Delta$-LLR wide, that should be concatenated to $\Lambda_N(f)$ for successful decoding of frame $f$. In inter-frame decoding, this means at least $\kappa_{(f)}$ subframes should be concatenated to $f$, for intra-frame decoding of $f$ to be successful. Therefore, the channel behavior in inter-frame decoding can be modeled as follows: \emph{it is characterized by a probability distribution over $\mathbb{Z^+}$, $(\delta_{(\omega)})_{\omega\geq 0}$, where $\delta_{(\omega)}$ is the probability that a frame $f$ transmitted over the channel has its $\kappa$-value $\kappa_{(f)}=\omega$.} To conclude, the model associates to each frame its appropriate code-rate, that is randomly sampled according to a probability distribution. This probability distribution depends on the channel characteristics, deployed intra-frame encoder/decoder and transmission scheduling. Therefore, its derivation will not be considered in this paper.

\textbf{Rate-matching process:} Under the developed model, inter-frame decoding can be projected as or, in other words, abstracted to a rate-matching process, in which the frame-length of each frame $f$ is increased to a value greater than or equal to $N+\Delta\cdot\kappa_{(f)}$. Therefore, the frame is assigned its appropriate rate determined by $\kappa_{(f)}$.
This process has two inputs: 1) the $N_F$ $\kappa$-values corresponding to the $N_F$ frames and 2) the binary representation $\mathbf{H_b}\!=\![h_{b(i,j)}]$ of $\mathbf{H}$, where $h_{b(i,j)}$ is set to $1$ if $h_{(i,j)}\!\neq\! 0$ and $0$ otherwise. The dependance of inter-frame decoding on $\mathbf{H_b}$ rather than on $\mathbf{H}$ is justified by assumption 1 stated in this subsection, where intra-frame decoding is assumed insensitive to the increment indices. For sake of illustration, the rate-matching procedure corresponding to the inter-frame decoding algorithm in subsection~\ref{s:DecProc} is described next. The definition and initialization of sets $\mathcal{R}$, $\mathcal{U}$, $\mathcal{P}$ are not altered and thus are not restated here. In the rate-matching procedure, a variable $\xi_{(f)}$ is associated to each frame $f$, where $\xi_{(f)}$  is the number of subframes concatenated to frame $f$; $\xi_{(f)}$ is initialized to 0. The procedure repeats the following iteration as long as $\mathcal{P}\neq\varnothing$:

{\bf Iteration:} If $\mathcal{P}\!=\!\varnothing$, quit decoding, else, pick randomly an element $f$ from $\mathcal{P}$. If $\xi_{(f)}<\kappa_{(f)}$, intra-frame decoding is unsuccessful, move $f$ to $\mathcal{U}$ and go back to the beginning of the step. If $\xi_{(f)}\geq \kappa_{(f)}$, intra-frame decoding is successful and $f$ is recovered, then proceed as follows:
\begin{itemize}
\item[1.] $f$ is moved from $\mathcal{P}$ to $\mathcal{R}$.
\item[2.] For every $s$ where $h_{b(s,f)}=1$:
 \begin{itemize}
  \item[a.] Set $h_{b(s,f)}$ to $0$.
  \item[b.] If the number of nonzero entries of row $s$ of $\mathbf{H_b}$ reduces to $1$, where $h_{b(s,f')}=1>0$, that is only one frame neighbor, $f'$, of $s$ is still unrecovered, then: 1) if $f'\in \mathcal{U}$, $f'$ is moved to $\mathcal{P}$, that is, frame $f'$ is rescheduled for intra-frame decoding, and 2) $\xi_{(f')}$ is incremented by $1$ to account for the concatenation of subframe $s$ to frame $f'$.
 \end{itemize}
\end{itemize}

\subsection{Two-phase Message-passing Procedure}\label{s:BiPartite}
\noindent The rate-matching procedure is mapped to a two-phase message-passing procedure applied on a bi-partite graph. The message-passing procedure has two features. First, it is equivalent to the rate-matching procedure in the following sense: it determines, for each frame, how much its code-rate will be decreased throughout the rate-matching process, and thus which frames will be recovered. Second, it is suitable for asymptotic analysis.

\textbf{Bi-partite graph description:} The inter-frame code can be described using a bi-partite graph as illustrated in Fig.~\ref{f:Bipartite}. The first partition $\mathcal{V}$ includes $N_F$ variable nodes, each corresponding to a frame, while the second partition $\mathcal{C}$ includes $K_S$ check nodes, each corresponding to a subframe. By abuse of notation, $v$/$c$ is used to index the corresponding frame/subframe. An edge exists between variable $v$ and check $c$ if the corresponding frame and subframe are neighbors; $v$ and $c$ are neighbor nodes. The degree of a node is the number of edges connected to it.
For each set of $N_F$ received frames, variable-node $v$ is associated to a value $\kappa_{(v)}$, sampled from $(\delta_{(\omega)})$.

\begin{figure}[hbtp]
\centering
\includegraphics[scale=0.6]{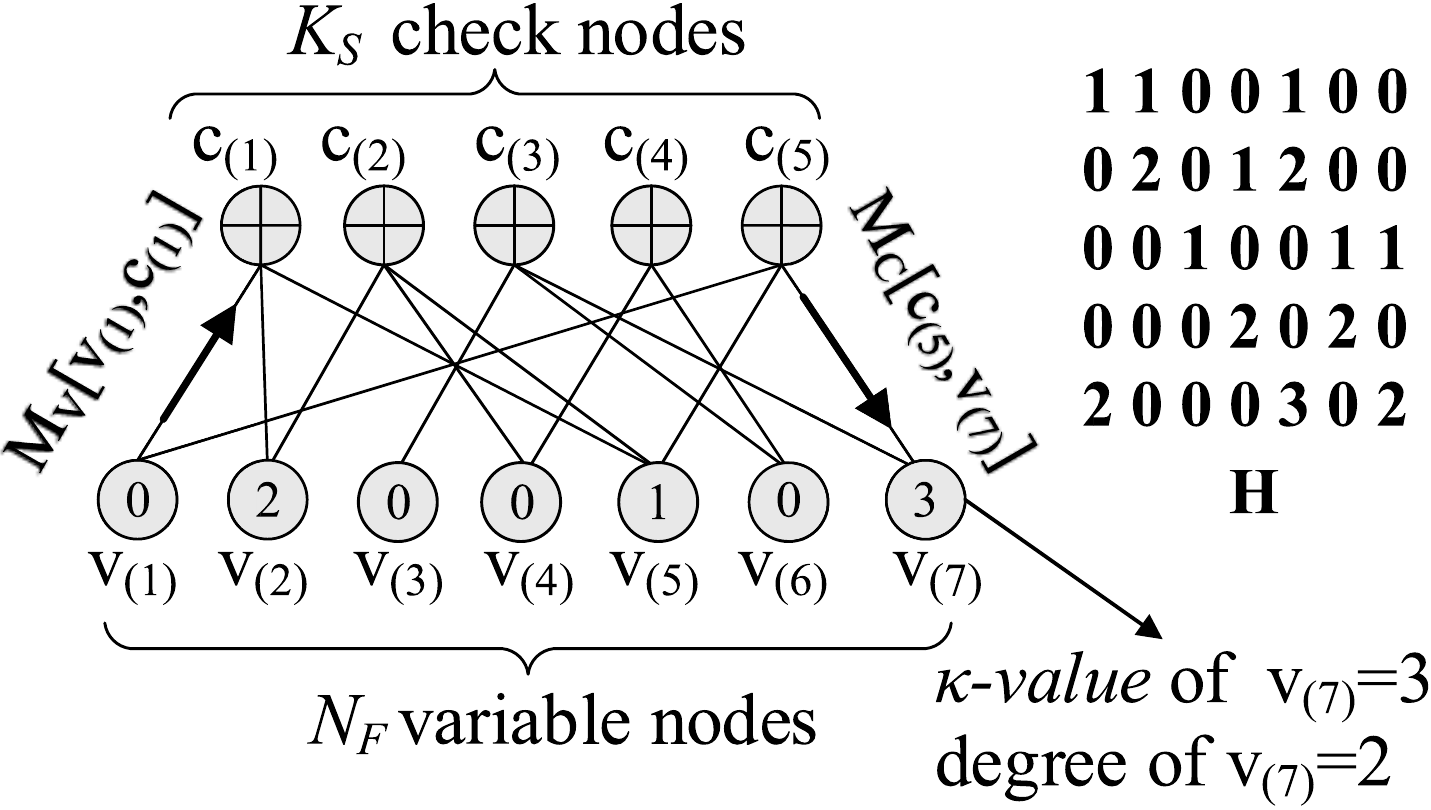}
\caption[Bipartite graph modeling of an inter-frame code]{Bipartite graph corresponding to an inter-frame code. Each variable node is indexed by its $\kappa$-value. Note that since the $\kappa$-value of node $v_{(7)}$ is greater than its degree, frame $7$ cannot be recovered.}
\label{f:Bipartite}
\end{figure}

\textbf{Two-phase message-passing procedure:} It is an iterative procedure applied on the bi-partite graph. Per iteration, along any edge $(v,c)$ connecting variable node $v$ and check node $c$, two binary messages are exchanged:  the variable-to-check message $M_V[v,c]$ and check-to-variable message $M_C[c,v]$. The procedure goes as follows: ($N_{(v)}/N_{(c)}$ is the set of neighbors of $v/c$)\\
\textbf{Initialization:} For every edge $(v,c)$ of the graph, $M_V[v,c]\leftarrow (\kappa_{(v)}==0)$; that is, the variable to check message, $M_V[v,c]$ is set to $1$ if $\kappa(v)$ is equal to $0$, and to $0$ otherwise.\\
\textbf{For every iteration $i=1\cdots MAXITER$:}
\begin{itemize}
\item \emph{Check-to-variable message update:} For every edge $(v,c)$ of the graph:
 \begin{equation}
 M_C[c,v]\leftarrow \prod_{v'\in \mathcal{N}_{(c)}\atop v'\neq v}M_V[v',c]
  \end{equation}
  That is, $M_C[c,v]$ is set to $1$ if all the variable-to-check messages incoming from all the neighbor nodes of $c$, excluding $v$, is $1$, and to $0$ otherwise.
\item \emph{Variable-to-check message update:}  For every edge $(v,c)$ of the graph:
 \begin{equation}
 M_V[v,c]\leftarrow \left(\kappa_{(v)}-\sum_{c'\in \mathcal{N}_{(v)} \atop c'\neq c}M_C[c',v]\leq0\right)
 \end{equation}
 That is, $M_V[v,c]$, is set to $1$ if the number of $1$-valued check-to-variable messages sent from all neighbors of $v$, excluding $c$, to $v$ is greater than or equal to $\kappa_{(v)}$.
\item \emph{Recovery flag update:} For every variable node $v$ of the graph:
 \begin{equation}
 U_{(v)}\leftarrow \left((\kappa_{(v)}-\sum_{c'\in \mathcal{N}_{(v)}}M_C[c',v])\leq0\right)
  \end{equation}
 That is, $U_{(v)}$ is set to $1$, if the number of $1$-valued check-to-variable messages sent from all neighbors of $v$ to $v$ exceeds $\kappa_{(v)}$. A value $1$ of $U_{(v)}$ indicates the frame corresponding to $v$ is recovered in the inter-frame decoding process.
\end{itemize}

The proof that this procedure determines the result of the rate-matching process and thus that of inter-frame decoding is intuitive, but the details are tedious. It will be omitted for brevity.

\textcolor{black}{
The described algorithm generalizes a two-phase message-passing algorithm that is applied in erasure decoding of LDPC codes and of the codes defined in~\cite{ErasureCodes,Practical_Loss_Resilient}. The erasure-decoding case is obtained by limiting the possible $\kappa$-values to $0$ and $1$, where $\delta_{(\omega)}=0$ for $\omega>1$. In this case, the following analogy can be made: 1) the bipartite graph described here can be viewed as representing an LDPC code, having $N_F$ variable nodes and $K_S$ check nodes, 2) for this code, the symbol corresponding to any variable node $v$ of this bipartite graph is considered erased if $\kappa_{(v)}=1$ and non-erased otherwise: $\delta_{(1)}$ represents the erasure rate, and 3) the algorithm developed here describes the progress of the iterative erasure decoding applied on the LDPC bipartite graph.
}

\subsection{Optimality Bound}\label{s:Bound}
\color{black}
\noindent In this subsection, the definition of optimality of the rate-matching process is developed.

For each distribution $(\delta_{(\omega)})$, the expected number of increments required to be concatenated to a single frame is $E(\kappa_{(f)})=\sum_\omega(\omega\cdot\delta_{(\omega)})$.

In inter-frame coding, the average number of transmitted, and therefore available, subframes per frame is $\frac{K_S}{N_F}$.
Besides, a subframe is concatenated to at most one frame throughout inter-frame decoding. Therefore, a necessary condition for this decoding to succeed to recover all frames is that $K_S\geq \sum_{f=1}^{N_F}\kappa_{(f)}$, or equivalently $\frac{K_S}{N_F}\geq \frac{\sum_{f=1}^{N_F}\kappa_{(f)}}{N_F}$.
Now, by the law of large numbers $\frac{\sum_{f=1}^{N_F}\kappa_{(f)}}{N_F}\rightarrow \sum_\omega(\omega\cdot\delta_{(\omega)})$, as $N_F\rightarrow\infty$.

A resulting question is then whether the rate-matching process can be applied successfully, to recover all or nearly all the $N_F$ frames with a probability close to $1$, while having: the ratio of \emph{the average number of available subframes per frame, $\frac{K_S}{N_F}$,} to\emph{ the expected number of required subframes per frame, $\sum_\omega(\omega\cdot\delta_{(\omega)})$,} approaching $1$ as $N_F\rightarrow \infty$. If the answer is positive, the rate-matching process achieves the lower bound on $\frac{K_S}{N_F}$ determined by $(\delta_{(\omega)})$: it is called \emph{optimal}.

Formal optimality definition: the rate-matching process is optimal if: for any distribution $(\delta_{(\omega)})$ and $\epsilon>0$, there exist inter-frame codes that can be defined over increasing values of $N_F$ such that: as $N_F\rightarrow \infty$, 1) $\frac{K_S}{N_F}\rightarrow \sum_\omega(\omega\cdot\delta_{(\omega)})$, and 2) the probability \emph{that a randomly-chosen frame is successfully-decoded in the inter-frame decoding process} is greater than $1-\epsilon$.

A relation between the definition of optimality made here and capacity achieving coding can be deduced from the link observed in subsection~\ref{s:BiPartite} between the rate-matching process and iterative LDPC erasure decoding.
When the special case characterized by having $\delta_{(\omega)}=0$ for $\omega>1$ is considered, the optimality of the rate matching process is equivalent to the following proposition: for an erasure channel of erasure rate $\delta_{(1)}$, there exist LDPC codes having rate $1-\frac{K_S}{N_F}\rightarrow 1-\delta_{(1)}$ as $N_F\rightarrow \infty$, and which iterative decoding succeeds when the erasure rate is $\delta_{(1)}$; this means these LDPC codes come arbitrarily close to achieving the capacity of an erasure channel with an erasure probability of $\delta_{(1)}$.

\color{black}

\subsection{Asymptotic Analysis}
\noindent A compact\footnote{The term ``compact", used here to describe a mathematical expression, means that the corresponding expression involves few well-defined terms.} mathematical expression that describes the progress of the message-passing procedure is derived in Theorem~\ref{t:Condition} of the subsection. This expression relates the degree-distributions, characterizing the bi-partite graph and thus the inter-frame code, to the distribution $(\delta_{(\omega)})$. It will
then be used to construct degree-distributions that describe inter-frame codes for which the rate-matching process achieves the optimal bound set in~\ref{s:Bound}.

\noindent \textbf{Assumption on $(\delta_{(\omega)})$:} an additional assumption is made here on $(\delta_{(\omega)})$, which is that $\delta_{(\omega+1)}=\mu\cdot\delta_{(\omega)}$, $\omega\geq1$, for some parameter $\mu<1$.  This means that $(\delta_{(\omega)})$
is fully described by two parameters, $\mu$ and $\delta=\sum_{\omega>0}\delta_{(\omega)}$, where: $\delta_{(0)}=1-\delta$ and $\delta_{(\omega)}=\delta\cdot(1-\mu)\cdot\mu^{\omega-1}$ for $\omega\geq1$.

\color{black}{
This geometric progression assumption is motivated by three factors. First, it simplifies the analysis done in this subsection, which, in turn, makes it easier to construct the optimal degree-distributions in the next subsection.
Second, it allows to apply a common practice in the field of mathematics to the analysis of inter-frame coding. This practice is
starting to solve a hard problem by solving special, and therefore simpler, cases of the problem. Applying this practice to the case of inter-frame coding means that a starting point to find whether the rate-matching process is optimal for all possible $(\delta_{(\omega)})$ distributions is to find whether this process is optimal for the special case of $(\delta_{(\omega)})$, characterized by the relation $\delta_{(\omega+1)}=\mu\cdot\delta_{(\omega)}$. No work has been done on the general case of $(\delta_{(\omega)})$ distribution in this paper though.\\
Third, the geometric progression assumption models $(\delta_{(\omega)})$ appropriately in some realistic scenarios. A simulation example is presented next to illustrate this point. The example is based on the following observation: within the channel variability model, when $i$ increments are concatenated to the $N$-bit portion of a frame, the probability that the intra-frame decoding of this frame fails equals $\sum_{\omega=i+1}^{\infty}\delta_{(\omega)}$. Now consider the curve of the intra-frame decoding-failure rate versus $i$, the number of concatenated increments:
that this curve is linear when the y-axis has a logarithmic scale is equivalent to $(\delta_{(\omega)})$ being a geometric progression. Fig.~\ref{f:ferplots_increments} shows two example curves of the intra-frame decoding-failure rate versus the number of concatenated increments, that are obtained through simulation as follows: for each integer $0\leq i\leq D$, the transmission and intra-frame decoding of a large number of frames of length $N+i\cdot\Delta$ are simulated. The setup for the transmission of each frame can be described as follows: a $K$-bit block is encoded via LTE turbo code~\cite{LTE_36.212} into a $(N\!+\!D\cdot\Delta\!<\!3K)$-bit encoded frame $f$ where $(K,N,\Delta)\!=\!(4096,K/0.75,0.1N)$. For every $i$, the transmission of the $(N\!+\!i\!\cdot\! \Delta)$-bit portion of the frame over a Pedestrian-B (PedB) fading channel~\cite{LTE_36.212}, subject to white Gaussian noise with fixed variance, is simulated using $16$-QAM modulation. It is assumed that the $N$-bit portion and each of the increments $\Delta(f,j)$, $1\!\leq\! j\!\leq\! i$, are transmitted over \emph{decorrelated} channel instances for curve $1$, and over the same channel instance for curve $2$. As shown in Fig.~\ref{f:ferplots_increments}, the curves obtained through simulation are nearly linear, justifying the assumption that the $(\delta_{(\omega)})$ corresponding to the simulated scenario is a geometric progression.

\begin{figure}[t]
\centering
\includegraphics[scale=0.37]{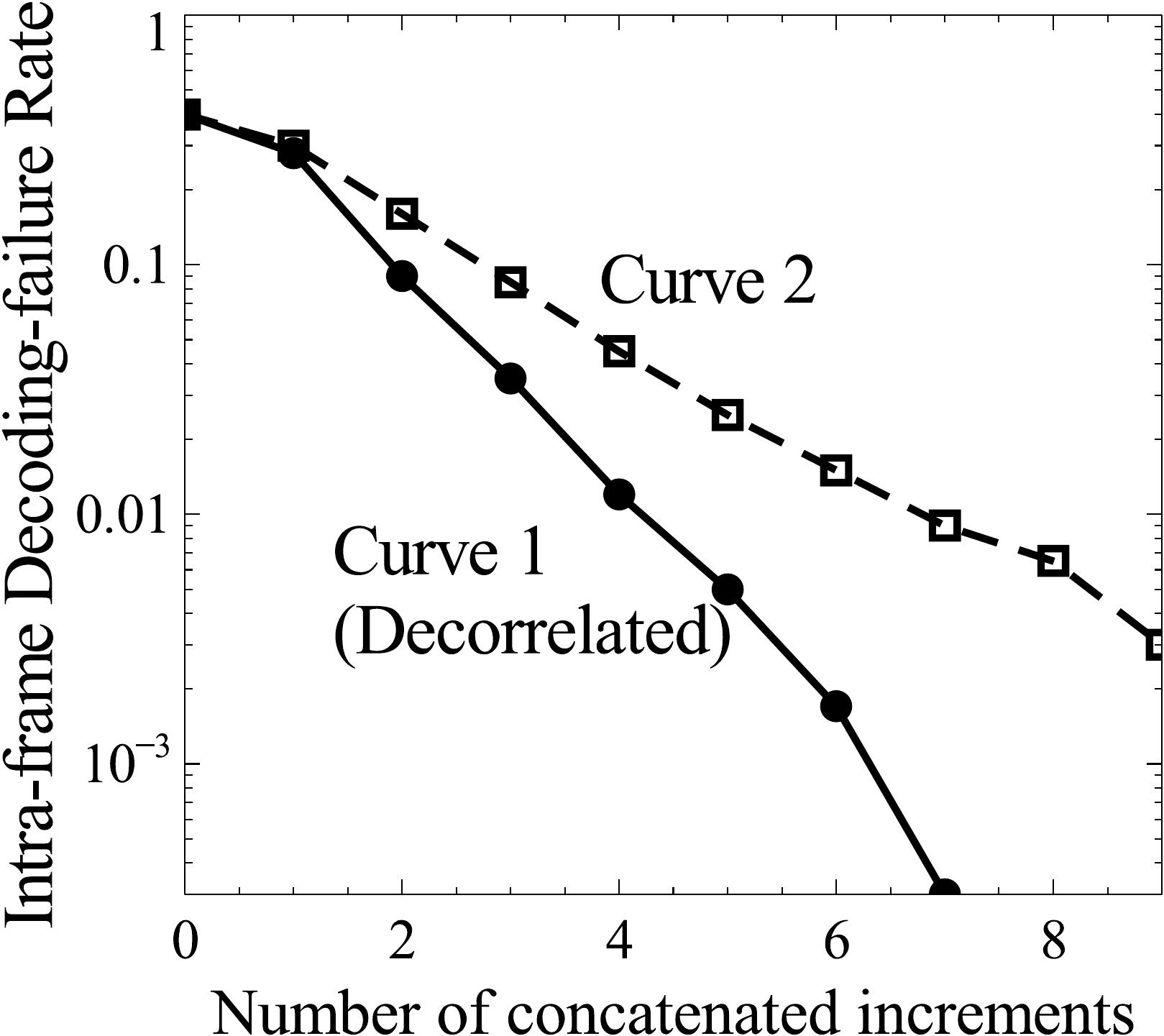}
\caption{Intra-frame decoding-failure rate versus the number of increments concatenated to the $N$-bit portion.}

\label{f:ferplots_increments}
\end{figure}

\color{black}

Since the developed message-passing algorithm generalizes that of erasure codes, the analysis performed here follows the technique of analysis of random processes via and-or tree evaluation in~\cite{AnalysisAndOr}. Therefore, the basic features of this technique are described first.

\noindent \textbf{Preliminaries:} The analyzed message-passing procedure is applied on bipartite graphs with the following properties: 1) $(N_F,K_S) \rightarrow (\infty,\infty)$, and 2) the graphs are constructed randomly according to two probability distributions: the edge variable and check degree-distributions, $(\lambda_{(i)})_{i\geq1}$ and $(\rho_{(i)})_{i\geq1}$ respectively. The following terminology is adopted from~\cite{ErasureCodes}: an edge in the bipartite graph drawn between variable node $v$ and check node $c$ has as variable degree the degree of $v$, and as check degree the degree of $c$. Then,  $\lambda_{(i)}(\rho_{(i)})$, $i\geq 1$, is the probability that a randomly-picked edge of the graph has variable(check) degree $i$. For any two distributions, $(\lambda_{(i)})_{i\geq1}$ and $(\rho_{(i)})_{i\geq1}$, random graphs with the respective edge degree-distributions can be constructed; these distributions characterize the constructed ensemble. The functions $\lambda(x)$ and $\rho(x)$ are defined as follows: $\lambda(x)=\sum_i\lambda_{(i)}\cdot x^{i-1}$ and $\rho(x)=\sum_i\rho_{(i)}\cdot x^{i-1}$.

\noindent \textbf{Analysis:} The message-passing procedure is applied on a randomly constructed bipartite graph. As in~\cite{AnalysisAndOr}, it is viewed as a random discrete process and the evolution of one of its parameters $Q_{(i)}$ throughout this process is analyzed, where $Q_{(i)}$ is defined as the probability \emph{that a randomly picked check-to-variable message at iteration $i$ is $1$}. Let $f(\cdot)$ be the function defined over $[0,1]$ such that $f(\cdot)$ maps $Q_{(i)}=y$ to $Q_{(i+1)}=f(y)$. Define $y^*\in[0,1]$ as follows:
\begin{equation}
f(y)>y  \indent   \forall 0\leq y<y^* \indent and \indent  f(y^*)=y^*.
\label{e:inequality1}
\end{equation}

The function $f(\cdot)$ is increasing; therefore, the convergence value of the sequence $\{Q_{(i)}\}$, $i=1\cdots\infty$, is $Q_{(\infty)}=y^*$. Inequality \eqref{e:inequality1} describes the progress of the message-passing procedure.

Let $g(x)$ be a function of $0\leq x \leq 1$, defined as the probability \emph{that a randomly-picked variable-to-check message is 0 }, given the probability \emph{that a randomly-picked incoming check-to-variable message is 0} is $x$, then $f(1-x)=\sum_i\rho_{(i)}.(1-g(x))^{i-1}=\rho(1-g(x))$. Similar to~\cite{AnalysisAndOr}, the latter equality uses the assumption the variable-to-check messages incoming to a check node of degree $i$ are independent random variables. Therefore, the probability \emph{that the outgoing check-to-variable message is $1$} is a product of $i-1$  identical values, each equal to $1-g(x)$ . Define $x^*=1-y^*$, inequality~\eqref{e:inequality1} can be reformulated to:
\begin{equation}
\rho(1-g(x))>1-x, \indent \forall x\in]x^*,1].
\label{e:inequality2}
\end{equation}

Inequality~\eqref{e:inequality2} is derived in~\cite{AnalysisAndOr}. The novel contribution of this subsection is shown next, where an inequality, specific to the message-passing algorithm developed in~\ref{s:BiPartite}, is derived.

\begin{theorem} (\emph{Message-passing procedure characterization})
For a distribution $(\delta_{(\omega)})$ described by $(\delta,\mu)$, and degree-distributions, $(\lambda_{(i)})_{i\geq1}$ and $(\rho_{(i)})_{i\geq1}$, the outcome of the message-passing procedure can be fully characterized by the inequality:
\begin{equation}
\rho(1-\delta\cdot\lambda(z)) > \frac{1-z}{1-\mu} \indent, \indent  \forall z\in]z^*,1].
\end{equation}
\label{t:Condition}
where $z^*=(1-y^*)+\mu \cdot y^*>\mu$.
\end{theorem}
\begin{proof}
Define $g_{(d;\omega)}(x)$ as the probability \emph{that a randomly-picked variable-to-check message is 0} given: 1) the corresponding variable-node $v$ has degree $d$, 2) $\kappa_{(v)}=\omega$ and 3) the probability \emph{that a randomly-picked incoming check-to-variable message is $0$} is $x$. Similarly, define $g_{(d)}(x)$ as the probability \emph{that a randomly-picked variable-to-check message is 0} given: 1) the corresponding variable-node $v$ has degree $d$ and 2) the probability \emph{that a randomly-picked incoming check-to-variable message is $0$} is $x$. The function $g(x)$ can be expressed as:
\begin{equation*}
g(x)=\sum_{d=1}^{\infty}\lambda_{(d)}\cdot g_{(d)}(x)   \indent, \indent   g_{(d)}(x)=\sum_\omega\delta_{(\omega)}\cdot g_{(d;\omega)}(x).
\end{equation*}
Then,
\begin{equation*}
g_{(d;\omega)}(x)=\sum_{j=0}^{\min(\omega,d)-1}\left(\begin{array}{c} d-1\\j\end{array}\right)\cdot(1-x)^j\cdot x^{d-1-j}.
\end{equation*}
\small
\begin{eqnarray*}
&&g_{(d)}(x)=\sum_\omega\delta_{(\omega)}\cdot g_{(d;\omega)}(x)\\
&=&\sum_{\omega=0}^{\infty}\delta_{(\omega)}\cdot \left(\sum_{j=0}^{\min(\omega,d)-1}\left(\begin{array}{c} d-1\\j\end{array}\right)\cdot(1-x)^j\cdot x^{d-1-j}\right)\\
&=& \sum_{j=0}^{d-1}\left(\left(\begin{array}{c} d-1\\j\end{array}\right)\cdot(1-x)^j\cdot x^{d-1-j}\cdot \sum_{\omega=j+1}^{\infty}\delta_{(\omega)}\right)\\
&=& \sum_{j=0}^{d-1}\left(\left(\begin{array}{c} d-1\\j\end{array}\right)\cdot(1-x)^j\cdot x^{d-1-j}\cdot \delta_{(1)}\cdot \mu^j\cdot \sum_{\omega=0}^{\infty}\mu^\omega \right)\\
&=& \left(\delta_{(1)}\cdot\sum_{\omega=0}^{\infty}\mu^\omega\right)\cdot\sum_{j=0}^{d-1}\left(\begin{array}{c} d-1\\j\end{array}\right)\cdot\mu^j\cdot(1-x)^j\cdot x^{d-1-j}\\
&=& \frac{\delta_{(1)}}{1-\mu}\cdot(x+\mu\cdot(1-x))^{d-1}=\delta\cdot(x+\mu\cdot(1-x))^{d-1}.
\end{eqnarray*}\normalsize

Therefore:
\begin{equation}
g(x)=\sum_{d=1}^{\infty}\lambda_{(d)}g_{(d)}(x)=\delta\cdot\lambda(x+\mu\cdot(1-x)).
\label{e:inequality33}
\end{equation}
Substituting the obtained expression of $g(x)$ in inequality (\ref{e:inequality2}), and applying the following change of variable $z=x+\mu(1-x)$, we get: for $z^*=(1-y^*)+\mu\cdot y^*\geq\mu$:
\begin{equation*}
\rho\left(\begin{array}{c}1-\delta\cdot\lambda(z)\end{array}\right)>\frac{1-z}{1-\mu}, \indent  z\in]z^*,1].
\end{equation*}
\end{proof}

The derived analytic expression is a generalization of that corresponding to erasure decoding~\cite{AnalysisAndOr}: by setting $\mu$ to $0$, the derived inequality reduces to $\rho(1-\delta\cdot\lambda(z))>1-z$, $1-y^*<z\leq1$. This result is not unexpected since the message-passing algorithm developed here is itself a generalized form of a LDPC erasure-decoding message-passing algorithm.

\subsection{Optimal Degree Distributions}
\noindent Based on the analysis in the previous subsection, degree distributions are constructed next, such that these distributions characterize inter-frame codes for which the rate-matching process succeeds in the sense defined in subsection~\ref{s:Bound}, and which ratio $\frac{K_S}{N_F}$ approaches the optimal lower bound $\sum_\omega(\omega\cdot\delta_{(\omega)})=\frac{\delta}{1-\mu}$.
Towards this aim, a sequence of degree-distribution couples parameterized by $J=1 \cdots \infty$, $\left((\lambda_{(i)})_{i\geq1},(\rho_{(i)})_{i\geq1}\right)^{(J)}$, is constructed such that:

\noindent 1. $\frac{K_S}{N_F}\rightarrow \sum_{\omega\geq1}(\omega\cdot\delta_{(\omega)})=\frac{\delta}{1-\mu}$, as $(J,N_F)\rightarrow \infty$.

\noindent 2. $\exists$ $0<Z_0<1$ and $J^\star \in \mathbb{Z^+}$, such that:
\begin{equation}
\rho^{(J)}(1-\delta\cdot\lambda^{(J)}(z)) > \frac{1-z}{1-\mu},\indent \forall J>J^\star \indent \text{and} \indent Z_0<z\leq1.
\end{equation}

\indent This inequality is that obtained in Theorem~\ref{t:Condition} of the previous subsection.

\noindent 3. The minimum variable-node degree goes to $\infty$ as $J\rightarrow \infty$.

Next, it is shown that if properties (1)-(3) are satisfied, then the rate-matching process is optimal in the sense formally described in subsection~\ref{s:Bound}.

Property (1) means that $\frac{K_S}{N_F}$ converges to the optimal value of $\sum_{\omega\geq1}(\omega\cdot\delta_{(\omega)})$ as $N_F\rightarrow \infty$.

Properties (2) and (3) imply that the degree distributions characterize inter-frame codes for which the inter-frame decoding succeeds in the sense defined in subsection~\ref{s:Bound}; this is  because for any $\epsilon>0$, there exists an integer $J_{(\epsilon)}\in\mathbb{Z^+}$, such that: the probability \emph{that a frame $f$ is unrecovered in the decoding of a random code constructed according to $\left((\lambda_{(i)})_{i\geq1},(\rho_{(i)})_{i\geq1}\right)^{(J)}$}, converges to a value less than $\epsilon$ as $N_F\rightarrow\infty$ if $J>J_{(\epsilon)}$.
This latter result is proved in the rest of this paragraph. Property (2) implies that, $\forall J>J^\star$, the value $z^*$ defined in the previous subsection is less than $Z_0$, or equivalently that the sequence $\{Q_{(i)}\}$ converges to a value $Q_{(\infty)}=y^*$ that is greater than $Y_0=\frac{1-Z_0}{1-\mu}$. Let $\tau$ be the minimum variable-node degree, i.e. $\lambda_{(i)}=0$ for $i<\tau$. Consider the probability \emph{that a randomly-picked variable-to-check message is 0} given 1) the corresponding variable-node $v$ has degree $\tau$ and 2) the probability \emph{that a randomly-picked check-to-variable message incoming to $v$ is $0$} is $1-y^*$: it is denoted by $g_{(\tau)}(1-y^*)$ in the proof of Theorem~\ref{t:Condition} and equals $\delta\cdot((1-y^*)+\mu\cdot y^*)^{\tau-1}=\delta\cdot((1-y^*\cdot(1-\mu))^{\tau-1}<\delta\cdot((1-Y_0\cdot(1-\mu))^{\tau-1}$. Property (3) means that $\tau$ goes to $\infty$ as $J\rightarrow \infty$, and therefore $g_{(\tau)}(1-y^*)<\delta\cdot((1-Y_0\cdot(1-\mu))^{\tau-1}$ approaches $0$ as $J$ goes to $\infty$. Furthermore, it can be checked that the probability \emph{that a frame corresponding to a randomly-picked variable node $v$ is unrecovered}, given the probability \emph{that a randomly-picked incoming check-to-variable message is $0$} is $x$, is less than $g_{(\tau)}(x)$. Therefore, as the probability \emph{that a randomly-picked check-to-variable message is $1$} converges to $y^*>Y_0$, the probability \emph{that a frame corresponding to a randomly-picked variable node is unrecovered} converges to $0$ as $(J,N_F)\rightarrow\infty$.

\noindent \textbf{Degree-distributions:} A sequence of degree-distributions is constructed next. Then, it is proved that this sequence satisfies properties (1)-(3). For sake of brevity, the derivations and proof details in the rest of this subsection are omitted.

The distributions constructed here are a generalization of the optimal distributions of the erasure codes in\cite{ErasureCodes}. The erasure code distributions are described briefly next. The $\lambda$-distribution is $\lambda_{(i)}=\frac{1}{H\cdot(i-1)}$, $i=2,\cdots,d+1$, for a chosen integer $d$ and $H=\sum_{i=1}^{d}\frac{1}{i}$. The average variable-node degree is $a_v=(\sum_i\frac{\lambda_{(i)}}{i})^{-1}=H\cdot(1+1/d)$.
The function $\lambda(x)\sim-\ln(1-x)/H$ (but strictly less). The $\rho$-distribution is described as follows: $\rho_{(i)}=\frac{e^{-\alpha}\cdot\alpha^{i-1}}{(i-1)!}$, $i=1,\cdots,\infty$. Thus, $\rho(x)=e^{\alpha\cdot(x-1)}$. The average check-node degree is $a_c=(\sum_i\frac{\rho_{(i)}}{i})^{-1}=\alpha/(1-e^{-\alpha})$.
Then, we have for erasure codes:
\begin{eqnarray}
\rho(1-\delta\cdot\lambda(x))&=&e^{-\alpha\cdot\delta\cdot\lambda(x)}
>e^{\frac{\alpha\cdot\delta}{H}\cdot\ln(1-x)}=(1-x)^{\frac{\alpha\cdot\delta}{H}}\nonumber\\
&\geq& (1-x), \indent \forall 0\leq x \leq1.
\end{eqnarray}
The last inequality is valid when $\frac{a_v}{a_c}\geq\delta\cdot(1+1/d)$.

\emph{Proposed Degree Distributions:} The proposed inter-frame code distributions are constructed next. The $\lambda$-distribution is parameterized by two integers: the parameter $J$ of the constructed distribution couple $\left((\lambda_{(i)})_{i\geq1},(\rho_{(i)})_{i\geq1}\right)^{(J)}$ and another integer $d$. It is assumed that as $J\rightarrow\infty$ so does $d$, however, no specific relation involving both of them is imposed. In the rest of the section, to simplify notation, the super-index $(J)$ will be omitted from the $\lambda$ and $\rho$ terms. \\
The constructed $\lambda$-distribution can be described as: ($H=\sum_{i=1}^di^{-1}$)
\begin{equation}
\lambda_{(J\cdot(i-1)+1)}=\frac{1}{H\cdot(i-1)}    , \indent   i=2,\cdots, d+1.
\end{equation}
and $\lambda_{(k)}=0$, otherwise.\\
The construction of the $\rho$-distribution involves two distributions. Define the distribution $(\beta_{(\alpha;i)})$, where:
\begin{equation}
\beta_{(\alpha;i)}=\frac{e^{-\alpha}\cdot\alpha^{i-1}}{(i-1)!}.
\end{equation}
This distribution is similar to the $\rho$-distribution in erasure codes. Define another distribution $(\Omega_{(i)})$, parameterized by the integer $d_c$, as such:
\begin{equation}
\Omega_{(i)}=\frac{1}{i\cdot H_c}, \indent i=1,\cdots,d_c.
\end{equation}
where $H_c=\sum_{i=1}^{d_c}\frac{1}{i}$. The parameter $d_c$ is chosen such that:$\sum_{i=1}^{d_c}\frac{1}{i}\leq J\cdot(1-\mu)<\sum_{i=1}^{d_c+1}\frac{1}{i}$. From the definition of $d_c$, it can be deduced that $\frac{J}{H_c}\geq\frac{1}{1-\mu}$ and that $|\frac{J}{H_c}-\frac{1}{1-\mu}|$ goes to $0$ as $J\rightarrow 0$. The $\rho$-distribution is formed as follows:
\begin{equation}
\rho_{(i)}=\sum_{j=1}^{d_c}\Omega_{(j)}\cdot\beta_{(j\cdot H/\delta;i)}.
\end{equation}

\noindent\textbf{Proof that properties (1)-(3) are satisfied:} Property (3) is satisfied because the minimum variable-node degree, obtained from a distribution $(\lambda_{(i)})^{(J)}$, is $J+1$ which clearly goes to $\infty$ as $J\rightarrow \infty$.

\noindent Property (1) is shown to be satisfied as stated in Lemma~\ref{l:prop1}.

\begin{lemma}
$\frac{K_S}{N_F}\rightarrow \frac{\delta}{1-\mu}$ as $(J,d)\rightarrow (\infty,\infty)$.
\label{l:prop1}
\end{lemma}
\begin{proof}
Each of the expressions $K_S\cdot a_c$ and $N_F\cdot a_v$ represents the number of edges in the bi-partite graph, therefore $K_S\cdot a_c=N_F\cdot a_v$ which implies that $\frac{K_S}{N_F}=\frac{a_v}{a_c}$.

The average variable-node degree is:
\begin{equation}
a_v=\left(\sum_{j=1}^{\infty}\frac{\lambda_{(j)}}{j}\right)^{-1}=J\cdot H\left(\sum_{i=1}^{d}\frac{1}{i\cdot(i+1/J)}\right)^{-1}.
\end{equation}

The average check-node degree is:
\begin{eqnarray*}
a_c &=&\left(\sum_{i=1}^{\infty}\frac{\rho_{(i)}}{i} \right)^{-1}=\frac{H_c\cdot H}{\delta}\cdot\left(\sum_{j=1}^{d_c}\frac{1-e^{-j\cdot H/\delta}}{j^2}\right)^{-1}.
\end{eqnarray*}

As $(J,d)\rightarrow\infty$,  $\frac{J}{H_c}\rightarrow \frac{1}{1-\mu}$ and $d_c\rightarrow\infty$.
Besides, both $\left(\sum_{i=1}^{d}\frac{1}{i\cdot(i+1/J)}\right)$ and $\left(\sum_{j=1}^{d_c}\frac{1-e^{-j\cdot H/\delta}}{j^2}\right)$ converge to the same finite value $\left(\sum_{i=1}^{\infty}\frac{1}{i^2}\right)$. Overall, $\frac{K_S}{N_F}=\frac{a_v}{a_c}\rightarrow\frac{\delta}{1-\mu}$.
\end{proof}

\noindent The proof that property (2) is satisfied is more complicated. As a prelude, the expressions of $\lambda(x)$, $\rho(x)$, and $\rho(1-\delta\cdot\lambda(x))$ are obtained as follows:
\begin{eqnarray*}
\lambda(x)&=&\sum_{i=2}^{d+1}\frac{x^{J\cdot(i-1)}}{H\cdot(i-1)}\sim (<) -\frac{1}{H}\cdot\ln (1-x^J).
\end{eqnarray*}
\begin{eqnarray*}
\rho(x)&=&\sum_{i=1}^{\infty}\rho_{(i)}\cdot x^{i-1}=\sum_{j=1}^{d_c}\frac{1}{j\cdot H_c}e^{\frac{j\cdot H}{\delta}\cdot(x-1)}.
\end{eqnarray*}
\begin{eqnarray}
\rho(1-\delta\cdot\lambda(x))&>&\rho(1+\frac{\delta}{H}\cdot\ln(1-x^J))\nonumber\\
&=&\sum_{j=1}^{d_c}\frac{1}{j\cdot H_c}e^{j\cdot\ln(1-x^J)}\nonumber
\\&=&\frac{1}{H_c}\cdot\sum_{j=1}^{d_c}\frac{(1-x^J)^j}{j}.\label{eq:Tx}
\end{eqnarray}

A rigourous proof that property (2) is satisfied is based on the following three Lemmas.
\begin{lemma}
If $(1-x^J)^{d_c}<(1-x)$ $\forall$ $0<X_l<x<1$,  then $\frac{1}{H_c}\cdot\sum_{i=1}^{d_c}\frac{(1-x^J)^i}{i}>\frac{1-x}{1-\mu}$
$\forall$ $X_l<x<1$.
\label{l:f_fprime}
\end{lemma}

\begin{lemma}
$\forall \eta>0$, $\exists J_0$ such that $e^{1-\mu-\eta}<d_c^{\frac{1}{J}}<e^{1-\mu+\eta}$, $\forall J>J_0$.
\label{l:dc}
\end{lemma}

\begin{lemma}
For each value $J\!\in\!\mathbb{Z^+}$, there exists a single value $X_c\in[0,1]$ such that $(1\!-\!x^J)^{d_c}\!>\!1\!-\!x$ if $x\!<\!X_c$, and $(1\!-\!x^J)^{d_c}\!<\!1\!-\!x$ otherwise. In addition, $X_c$ converges to $e^{-(1-\mu)}$ as $J\!\rightarrow\!\infty$.
\label{l:convergence_point}
\end{lemma}

From Lemma~\ref{l:convergence_point}, it can be concluded that for an arbitrarily small $\gamma$, $\exists J^\star$ such that:
$X_c<e^{-(1-\mu)+\gamma} \indent \forall J>J^\star$. Therefore:
\begin{equation}
(1-x^J)^{d_c}<1-x \indent\indent \forall x>e^{-(1-\mu)+\gamma} \indent \text{and}\indent J>J^\star.
\end{equation}
This means, by Lemma~\ref{l:f_fprime}, that $\forall x>e^{-(1-\mu)+\gamma}$ and $J>J^\star$:
\begin{equation}
\frac{1}{H_c}\cdot\sum_{i=1}^{d_c}\frac{(1-x^J)^i}{i}>\frac{1-x}{1-\mu}.
\label{eq:Fx}
\end{equation}
Define function $T(x)=\frac{1}{H_c}\cdot\sum_{i=1}^{d_c}\frac{(1-x^J)^i}{i}$. By inequality~\eqref{eq:Fx}, and since $\rho(1-\delta\cdot\lambda(x))>T(x)$ by~\eqref{eq:Tx}, property (2) is satisfied.

It should be noted, that throughout the proof of the Lemmas, a stronger result can be derived on the function $T(x)$, which is:
\begin{equation*}
T(x)\rightarrow \left\{\begin{array}{ll}
                   1, & \hbox{$x<e^{-(1-\mu)}$} \\
                   -\frac{\ln(x)}{1-\mu}, & \hbox{$x\geq e^{-(1-\mu)}$}
             \end{array}\right.    \indent\text{as}\indent (J,d)\rightarrow(\infty,\infty)
\end{equation*}

\subsection{Conclusions on Data-Rate}
\noindent Overall, it is proved in the section that for a distribution $(\delta_{(\omega)})$ described by $(\delta,\mu)$ , the minimum required $\frac{K_S}{N_F}$ equals $\frac{\delta}{1-\mu}$, and therefore the effective frame-length is:

\begin{equation*}
N\cdot(1+\frac{\delta}{1-\mu}\cdot\frac{\Delta}{N}).
\end{equation*}
\section{Data-rate Comparison}\label{s:Comparison}
\textcolor{black}{
\noindent The data-rates achievable by inter-frame coding are compared to those achievable by two other schemes: 1) the state-of-the-art two-stage scheme, and 2) a simple frame-wise feedback-based rate-matching scheme.
This latter scheme is chosen for comparison because it shares with IR-HARQ a basic property: that the rate of a frame is decreased upon feedback by sending successive increments until the number of receivers which succeed to decode the frame is above a predefined threshold. For accurate comparison between inter-frame coding and IR-HARQ, many details on the involved communication scenario have to be specified. These details include the type and frequency of the involved feedback\footnote{The feedback may include, besides feedback on the decoding success or failure, regular updates on the channel-state for each receiver}, the time requirements of the communication, the transmission scheduling details, and other channel characteristics; such issues are beyond the scope of this paper.
}

The comparison is done by computing the enhancement ratio, defined as the ratio of the effective frame-length obtained in the conventional solution to that obtained in inter-frame coding. The comparison done here assumes the channel variability model developed in~\ref{s:Model}. An underlying assumption made in this section is that the distribution $(\delta_{(\omega)})$ does not change with the different schemes compared. Therefore, the reported results can be fully attributed to the peculiar coding features of these schemes.

It is assumed throughout this section that all the receivers, in the broadcast communication, have the same channel-characterizing distribution $(\delta_{(\omega)})$. This \emph{single-distribution} assumption is motivated by two considerations.
First, the assumption keeps the comparison setup within the range of the results obtained in the previous section: the previous analysis does not deal with the performance of an inter-frame code under multiple  $(\delta,\mu)$ couples. Second, the \emph{single-distribution} assumption  can be viewed as a simplification of the \emph{worst-distribution} assumption explained as follows: among the channel-characterizing probability distributions corresponding to the different receivers, there exist a single distribution which has the highest values of both $\delta$ and $\mu$.
The design of the inter-frame code and the two-stage scheme are determined solely by this distribution. It is, therefore, the distribution considered in evaluating the effective frame-length resulting from inter-frame coding, as well as, from the two-stage scheme.

In this regard, in comparing the inter-frame coding to the two-stage scheme, the number of receivers involved in the communication scenario are overlooked. The reason is that the coding-performance for each of the two compared approaches, per receiver, does not depend on the number of receivers. In contrast, the comparison between inter-frame coding and the frame-wise feedback-based rate-matching scheme is done with respect to the number of receivers involved.

\color{black}

\subsection{Frame-wise Feedback-based Rate-matching}\label{s:IR-HARQ}
\noindent A simple frame-wise feedback-based rate-matching scheme is assumed here where for each frame $f$, the rate-$R_H$ $N$-bit portion of the frame is initially transmitted. A retransmission is initiated as long as: 1) feedback from at least one receiver reports a decoding failure and 2) the number of retransmissions corresponding to $f$ are less than some upper bound $n^\star$.

The $i\text{th}$ retransmission, $1\leq i\leq n^{\star}$, consists of increment $\Delta(f,i)$.
For this subsection, the following terminology is defined: $n_{(f)}$ is the number of transmitted increments for frame $f$, $\mathcal{R}$ is the set of receivers,
and $\kappa_{(f;i)}$, $1\leq i\leq |R|$, is the $\kappa$-value of frame $f$ for the $i\text{th}$ receiver. It can be deduced from the description of the frame-wise feedback-based scheme that: 1) $n_{(f)}=\min(\max_{1\leq i\leq|\mathcal{R}|}\kappa_{(f;i)},n^\star)$, 2) and the frame-error-rate (FER) for each receiver equals $\sum_{\omega>n^\star}\delta_{(\omega)}=\delta\cdot\mu^{n^\star}$.
For a target FER of $10^{-a}$, $n^\star$ is set to $\lceil\frac{-a-\log_{10}(\delta)}{\log_{10}(\mu)}\rceil$.

For $1 \leq n<n^\star$, the probability that $n_{(f)}=n$ is equal to
$(1-\mu^{n}\cdot\delta)^{\mathcal{|R|}}-(1-\mu^{n-1}\cdot\delta)^{\mathcal{|R|}}$. Therefore, the expected value of $n_{(f)}$
is:
\small
\begin{equation}
E(n_{(f)})=n^\star-\sum_{i=0}^{n^\star-1}\left(1-\mu^i\cdot\delta\right)^{\mathcal{|R|}}=\sum_{i=0}^{n^\star-1}\left(1-\left(1-\mu^i\cdot\delta\right)^{\mathcal{|R|}}\right).
\end{equation}\normalsize
Expectedly, as $\mathcal{|R|}\rightarrow\infty$, $E(n_{(f)})\rightarrow n^\star$. The enhancement ratio is equal to:
\small
\begin{equation*}
\frac{1+E(n_{(f)})\cdot\frac{\Delta}{N}}{1+\frac{\delta}{1-\mu}\cdot\frac{\Delta}{N}}.
\end{equation*}
\normalsize
Fig.~\ref{f:Feedback-IF} shows the variation of the enhancement ratio versus $\mu$, for: $\frac{\Delta}{N}=0.1$, $\delta=0.5$, $\mathcal{|R|}\in\{10,30,100,\infty\}$, and target FER values of $10^{-1}$, $10^{-2}$, and $10^{-3}$. The enhancement ratios
grow with $\mathcal{|R|}$. For $\mathcal{|R|}=\infty$, the enhancement ratio values are considerably high: for example, for a target FER of $10^{-2}$, the enhancement ratio exceeds $2$ for $\mu\rightarrow0.8$, and exceeds $3$ for $\mu\rightarrow 0.9$. Yet, the enhancement brought by inter-frame coding over the frame-wise feedback-based scheme is significant, even when the number of receivers is small. For example, it can be checked that the enhancement ratio for a target FER of $10^{-2}$ and $\mathcal{|R|}=10$  exceeds the corresponding enhancement ratio, for a target FER of $10^{-1}$ and $\mathcal{|R|}\rightarrow\infty$, for all values of $\mu$.

\begin{figure*}[t]
\centering
\includegraphics[scale=0.34]{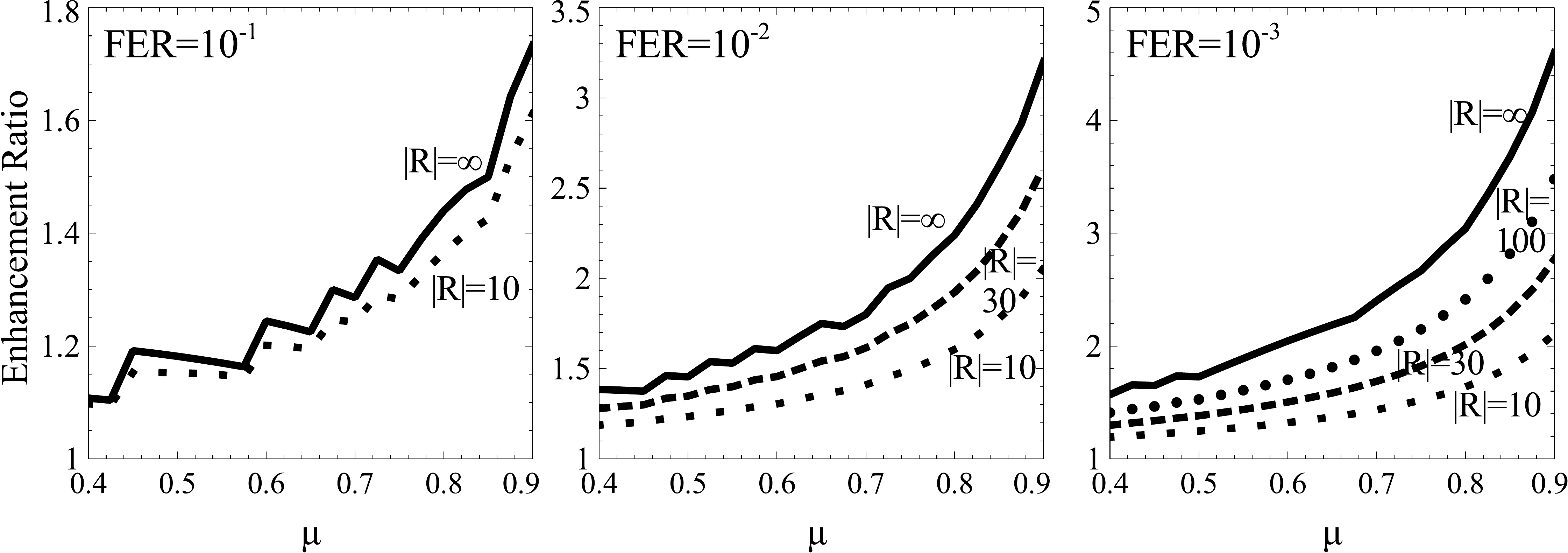}
\caption{Enhancement ratio, for different values of $\mathcal{|R|}$ and target FER. The fluctuations in the enhancement ratio curves in the first figure is due to the fact that $n^\star$ changes in steps of $1$ with $\mu$. }

\label{f:Feedback-IF}
\end{figure*}

\color{black}
\subsection{Two-stage Scheme}\label{s:two-stage}

\noindent In the comparison done here, the distribution $(\delta_{(\omega)})$ is assumed to be known to the sender.
Besides, for simplicity, no further feedback from the receiver to the transmitter is assumed.

The two-stage scheme follows the setup described in section~\ref{s:Setup}. On the transmitter side, the $N_F\cdot K$ information bits are encoded into $N_T\cdot K$ bits, using an erasure code of rate $R_E=\frac{N_F}{N_T}$. The resulting $(N_T\cdot K)$-bit sequence is then partitioned into $N_T$ $K$-bit blocks that are intra-frame encoded into $N_T$ rate-$R$ frames that are transmitted over the channel. On the receiver side, intra-frame decoding is applied on each received frame: if decoding fails, the frame is dropped and considered erased; else, if decoding succeeds, the $K$ bits corresponding to the frame are recovered and collected. The collected bits, resulting from the intra-frame decoding of the $N_T$ frames, are then forwarded into erasure decoding. The erasure code is capacity-achieving. Therefore, if slightly more than $N_F$ frames, out of a total of $N_T$ frames, are successfully intra-frame decoded, enough bits ($>N_F\cdot K$) are collected and erasure decoding succeeds.

The combination of the intra-frame code-rate and the erasure code-rate that results in the lowest effective frame-length should be found. This problem is considered in\cite{Rate_Det_1,Rate_Det_2,Rate_Det_3} under various assumptions on the channel model and communication scenarios. It is considered here again, under the channel variability model developed in subsection~\ref{s:Model}. For the asymptotic case, that is when $N_F\rightarrow \infty$, the optimal rate of the erasure-code $R_E$ is $1-\Gamma$, $\Gamma$ being the intra-frame decoding-failure rate. If the intra-frame code-rate $R$ is set to $\frac{K}{N+i\cdot\Delta}$, $i\in\mathbb{Z^{+}}$, that is the frame-length is set to $N+i\cdot\Delta$, $\Gamma$ is equal to $\sum_{\omega=i+1}^{\infty}\delta_{(\omega)}=\delta\cdot\mu^{i}$. It is reasonably assumed that for the general case of $i\in\mathbb{Q^+}$, $\Gamma$ is also equal to $\delta\cdot\mu^i$.
Therefore, the effective frame-length of the two-stage scheme when $(R,R_E)=(\frac{K}{N+i.\cdot\Delta},1-\delta\cdot\mu^i)$ is denoted here by $L_{(i)}$ and is equal to:
\begin{equation}
L_{(i)}=\frac{N+i\cdot\Delta}{R_E}=\frac{N+i\cdot\Delta}{1-\delta\cdot\mu^i}.
\end{equation}
The effective-frame length of the two-stage scheme is then the minimum, over $i\in\mathbb{Q^+}$, of $L_{(i)}$. It should be noted that $i$ is restricted to positive values to rule out the unrealistic case of having $R=\frac{K}{N+i\cdot\Delta}>1$ if $i$ is negative given that the exact value $\frac{K}{N}$ is not set here: $N$ is here assumed to be the minimum frame-length for both inter-frame coding and the two-stage scheme.

The enhancement ratio is then equal to:
\begin{equation}
\min_{i\in\mathbb{Q^+}} \frac{L_{(i)}}{(N+\frac{\delta}{1-\mu}\cdot\Delta)}= \min_{i\in\mathbb{Q^+}}\frac{1+i\cdot\frac{\Delta}{N}}{(1-\delta\cdot\mu^i)\cdot(1+\frac{\delta}{1-\mu}\cdot\frac{\Delta}{N})}.
\label{e:ratio}
\end{equation}

\begin{lemma}
For the two-stage-scheme, the enhancement ratio is upper bounded by $\frac{1}{1-e^{-1}}=1.582$.
\label{l:bound}
\end{lemma}

\begin{proof}
\begin{eqnarray*}
&&\min_{i\in\mathbb{Q^+}}\frac{1+i\cdot\frac{\Delta}{N}}{(1-\delta\cdot\mu^i)\cdot(1+\frac{\delta}{1-\mu}\cdot\frac{\Delta}{N})}\\
&\leq&\left(\frac{1+i\cdot\frac{\Delta}{N}}{(1-\delta\cdot\mu^i)\cdot(1+\frac{\delta}{1-\mu}\cdot\frac{\Delta}{N})}\right)_{i=\frac{\delta}{1-\mu}}\\
&=&(1-\delta\cdot\mu^{\frac{\delta}{1-\mu}})^{-1}.
\end{eqnarray*}

Consider the function $f(x)=x\cdot\mu^\frac{x}{1-\mu}$ over the range $[0,1]$. By simple calculus, it can be checked that it attains its maximum at $x^*=-\frac{1-\mu}{\ln(u)}$:
\begin{eqnarray*}
f(x^*)&=&-\frac{1-\mu}{\ln(\mu)}\mu^{-\frac{1}{\ln(u)}}
=-\frac{1-\mu}{\ln(\mu)}\cdot e^{-1}< e^{-1}.
\end{eqnarray*}
Therefore, $(1-\delta\cdot\mu^{\frac{\delta}{1-\mu}})^{-1}=(1-f(\delta))^{-1}<(1-e^{-1})^{-1}$.
\end{proof}

The curves plotted in Fig.~\ref{f:CodingOverhead} show $\frac{L_{(i)}}{(N+\frac{\delta}{1-\mu}\cdot\Delta)}$ versus $i$. The curves correspond to all possible combinations of $\delta\in\{0,3,0.5,0.8\}$ and $\mu\in\{0.5,0.7,0.85\}$, for $\frac{\Delta}{N}=0.1$. For all the considered $(\delta,\mu)$ pairs, except one ($(\delta,\mu)=(0.3,0.85)$), we have: the minimum of $\frac{L_{(i)}}{(N+\frac{\delta}{1-\mu}\cdot\Delta)}$ occurs for strictly positive values of $i$. This implies that restricting $i$ to positive values has no impact on the computed enhancement ratios, particularly when these ratios are relatively high.
In addition, the enhancement ratio $\min_{i\in\mathbb{Q^+}} \frac{L_{(i)}}{(N+\frac{\delta}{1-\mu}\cdot\Delta)}$, in general, grows as the channel statistical parameters ``worsens", that is as $\delta$ and $\mu$ increase. For example, while the enhancement ratio is $\sim 1.2$
for $(\delta,\mu)=(0.3,0.5)$ and $\sim1.25$ for $(\delta,\mu)=(0.5,0.5)$, it is  $\sim 1.35$ for $(\delta,\mu)=(0.5,0.7)$, $\sim 1.4$ for $(\delta,\mu)=(0.5,0.85)$, and $\sim 1.5$ for $(\delta,\mu)=(0.8,0.85)$.

\begin{figure*}[t]
\centering
\includegraphics[scale=0.34]{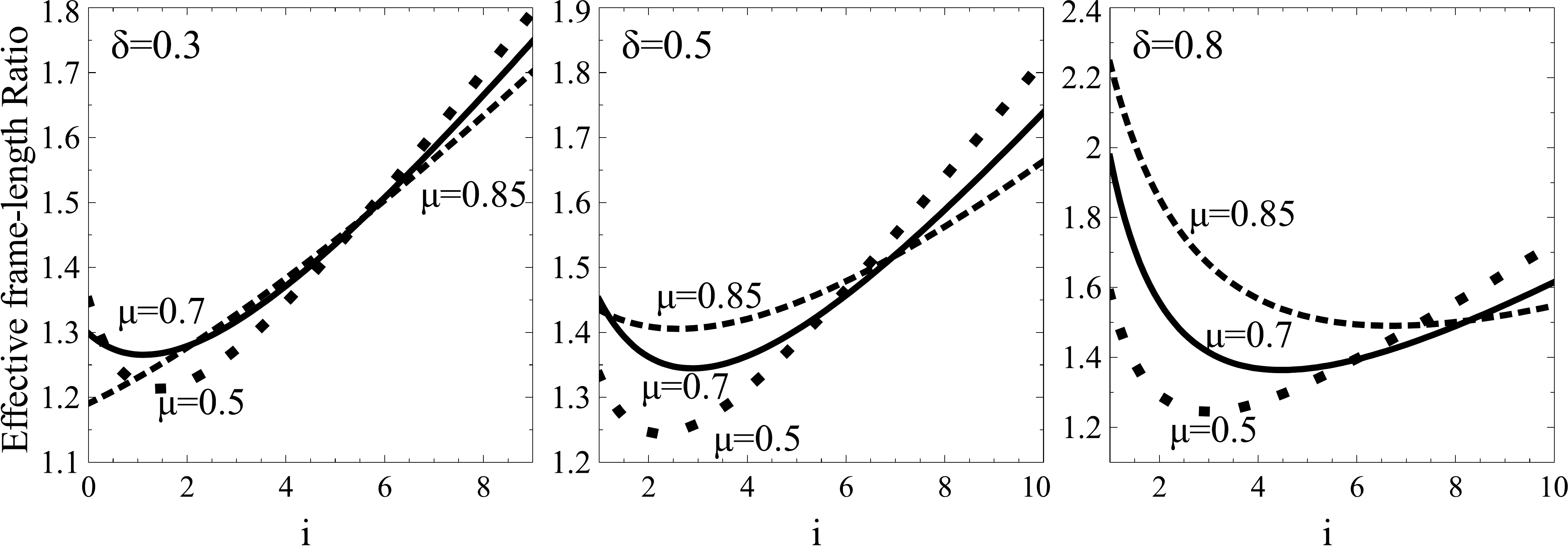}
\caption[\emph{Effective frame-length} ratio versus intra-frame length, for different values of $\delta$ and $\mu$]{\emph{Effective frame-length} ratio versus $i$, for different values of $\delta$ and $\mu$.}
\label{f:CodingOverhead}
\end{figure*}

The variation of the enhancement-ratio with the value of $\mu$ is illustrated in Fig.~\ref{f:CodingPerformance_Delta_1}, assuming $\frac{\Delta}{N}=0.1$. Three curves are plotted. One curve corresponds to a fixed value of $\delta=0.99$; the other two curves assume $\delta$ varies with $\mu$ such that $\delta=\mu^3$ and $\delta=\mu^6$ respectively. The assumption on how the two parameters of the distribution $(\delta_{(\omega)})$, $\delta$ and $\mu$, are related is justified as follows: the intra-frame decoding-failure rate is $\delta$ when the frame-length is $N$, and $\delta\cdot\mu^{\frac{N}{\Delta}}=\delta\cdot\mu^{10}$ when the frame-length is $2\cdot N$. If the same $N$-bit frame is transmitted twice under independent channel-state instances and each frame is independently intra-frame decoded, the probability of recovering the frame is $1-\delta^2$ for total number of sent bits of $2\cdot N$.  This probability is typically less than the probability of intra-frame decoding-success of a $(2\cdot N)$-bit frame, formed by concatenating $10$ subframes to the $N$-bit frame; therefore, $1-\delta^2<1-\delta\cdot \mu^{10}$ implying $\delta>\mu^{10}$. Therefore, the assumed relations of $\delta=\mu^3$, and $\delta=\mu^6$, consist merely example relations that satisfy the condition $\delta>\mu^{10}$ over the full considered range of $\mu$.
Three observations can be made from the figure. First, the derived upper bound of $\frac{1}{1-e^{-1}}$ on the enhancement ratio is relatively tight: the maximum value of the plotted enhancement ratios is $1.55$ and it occurs when $(\delta,\mu)=(0.99,0.94)$. Besides, even when $\delta$ varies with $\mu$, such that $\delta=\mu^3$, the enhancement ratio exceeds $1.5$ for $0.89\leq\mu\leq0.95$. Second, the variation of the enhancement ratio with $\mu$, for the three curves, follows the same trend. The enhancement ratio increases as $\mu$ increases to a value near $0.94$, at which the maximum is attained, and then it decreases as $\mu$ goes from $0.95$ to $1$. The rate of the increase in the plotted enhancement ratio is such that the enhancement ratio becomes relatively high as $\mu$ exceeds $0.75$ or $0.8$. Third, for the same value of $\mu$, higher enhancement ratios are obtained for higher $\delta$ values. Overall, the observed trend that the enhancement ratio increases with increasing $\delta$ and $\mu$ can be qualitatively justified as follows: as $\delta$ and $\mu$ increase, the effective-frame length increases in both schemes. That is, the value of the effective frame-length minus $N$ increases, relative to $N$. Consequently, the reduction in this value achieved by inter-frame coding compared to the two-stage scheme becomes more significant in the computation of the enhancement ratio.

\begin{figure}[t]
\centering
\includegraphics[scale=0.37]{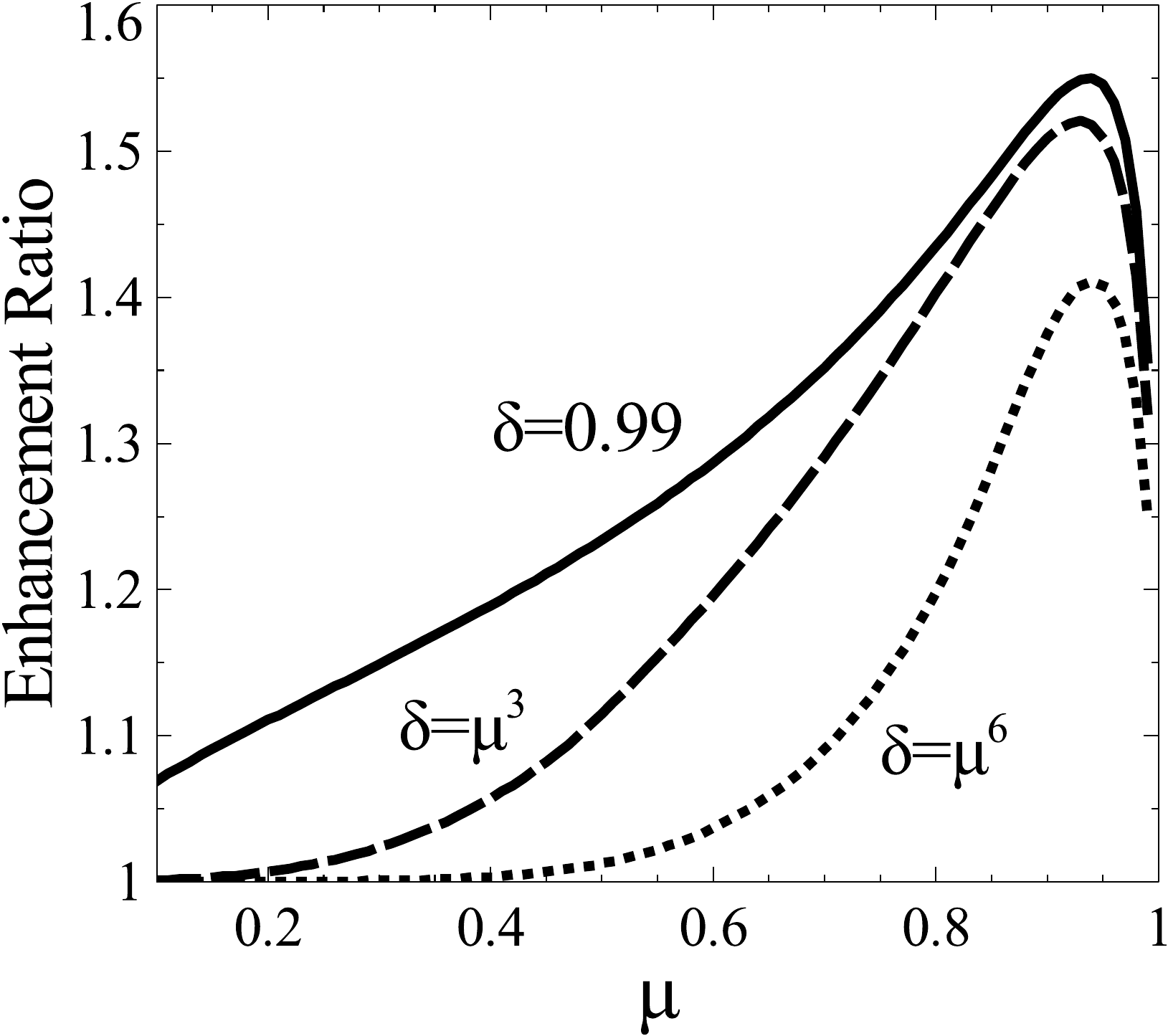}
\caption[Enhancement ratio compared to the two-stage scheme, versus $\mu$]{Enhancement ratio curves versus $\mu$, for different values of $\delta$.}
\label{f:CodingPerformance_Delta_1}
\end{figure}

\section{Conclusions}\label{s:Conclusions}
\noindent In this paper, an inter-frame coding approach to the problem of varying channel conditions has been developed. The proposed approach involves a rate-matching process, in which each frame is assigned its appropriate encoding-rate at the receiver side. This makes inter-frame coding suitable for broadcast communication, as well as for any communication scenario in which feedback from the receiver(s) to the transmitter is costly or undesirable. In terms of complexity, the computational overhead of inter-frame decoding includes operations that are similar in type and scheduling to those employed in the relatively-simple iterative erasure decoding. Through channel modeling, graph-based simplification and subsequent analysis, it has been shown in this paper that the matching process can be optimal under appropriate randomized construction of bi-partite graphs. Furthermore, the potential of inter-frame coding to enhance the achieved data-rates in broadcast communication has been quantitatively illustrated. This research direction can be extended in several directions. One direction is on the implementation aspects of the proposed scheme. Another is on the design of efficient moderate-size inter-frame codes (e.g. $N_F~\sim 10^2-10^3$). A third is on the performance of inter-frame coding when $(\delta_{(\omega)})$ is not described by a geometric progression, or when different receivers have different $(\delta_{(\omega)})$ distributions.

%% file: Author1_Biography.txt
\begin{IEEEbiography}[{\includegraphics[width=1in,height=1.15in,clip]{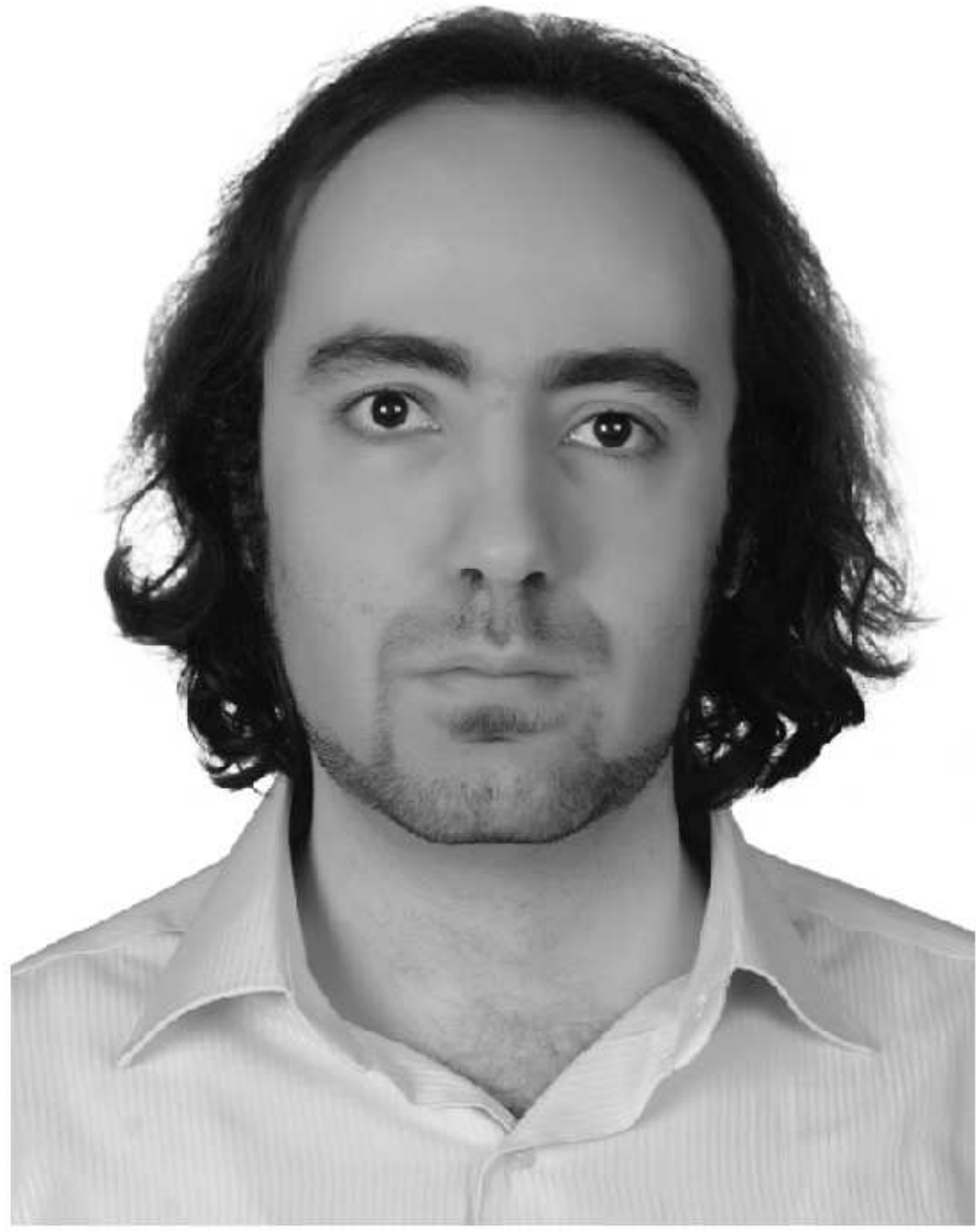}}]{Hady Zeineddine}
received the B.E. degree from the American
University of Beirut (AUB), Lebanon,
the M.S. degree from the University of Texas at Austin, and the PH.D. degree from the American University of Beirut, in 2006, 2009, and 2015 respectively. His research interests lie in the field of coding, including theory and applications, and in the design of algorithms and architectures for efficient IC
implementation of communication and digital signal processing
applications.

\end{IEEEbiography}

%% file: Author2_Biography.txt
\begin{IEEEbiography}[{\includegraphics[width=1in,height=1.25in,clip,keepaspectratio]{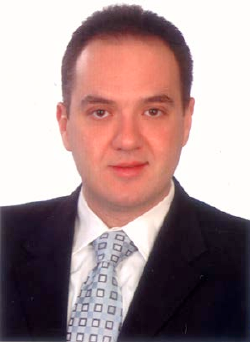}}]{Mohammad M. Mansour}
Mohammad M. Mansour (S’97–M’03–SM’08) received the B.E. (Hons.) and
the M.E. degrees in computer and communications engineering from
the American University of Beirut (AUB), Beirut, Lebanon, in 1996
and 1998, respectively, and the M.S. degree in mathematics and the
Ph.D. degree in electrical engineering from the University of
Illinois at Urbana–Champaign (UIUC), Champaign, IL, USA, in 2002
and 2003, respectively.

He was a Visiting Researcher at Broadcom, Sunnyvale, CA, USA, from
2012 to 2014, where he worked on the physical layer SoC
architecture and algorithm development for LTE-Advanced. He was on
research leave with Qualcomm Flarion Technologies in Bridgewater,
NJ, USA, from 2006 to 2008, where he worked on modem design and
implementation for 3GPP-LTE, 3GPP2-UMB, and peer-to-peer wireless
networking physical layer SoC architecture and algorithm
development. He was a Research Assistant at the Coordinated Science
Laboratory (CSL), UIUC, from 1998 to 2003. He worked at National
Semiconductor Corporation, San Francisco, CA, with the Wireless
Research group in 2000. He was a Research Assistant with the
Department of Electrical and Computer Engineering, AUB, in 1997,
and a Teaching Assistant in 1996. He joined as a faculty member
with the Department of Electrical and Computer Engineering, AUB, in
2003, where he is currently a Professor. His research interests are
in the area of energy-efficient and high-performance VLSI circuits,
architectures, algorithms, and systems for computing,
communications, and signal processing.

Prof. Mansour is a member of the Design and Implementation of
Signal Processing Systems (DISPS) Technical Committee Advisory
Board of the IEEE Signal Processing Society. He served as a member
of the DISPS Technical Committee from 2006 to 2013. He served as an
Associate Editor for IEEE TRANSACTIONS ON CIRCUITS AND SYSTEMS II
(TCAS-II) from 2008 to 2013. He currently serves as an Associate
Editor of the IEEE TRANSACTIONS ON VLSI SYSTEMS since 2011, and an
Associate Editor of the IEEE SIGNAL PROCESSING LETTERS since 2012.
He served as the Technical Co-Chair of the IEEE Workshop on Signal
Processing Systems in 2011, and as a member of the Technical
Program Committee of various international conferences and
workshops. He was the recipient of the PHI Kappa PHI Honor Society
Award twice in 2000 and 2001, and the recipient of the Hewlett
Foundation Fellowship Award in 2006. He has six issued U.S.
patents.

\end{IEEEbiography}